\documentclass[10pt]{article}

\usepackage[utf8]{inputenc}

\usepackage{epsf}
\usepackage{amsmath}

\allowdisplaybreaks

\usepackage[showframe=false]{geometry}
\usepackage{changepage}

\usepackage{epsfig}
\usepackage{amssymb}

\usepackage{amsthm}
\usepackage{setspace}
\usepackage{cite}
\usepackage{mcite}

\usepackage{algorithmic}  % This package provides an algorithmic environment fo describing algorithms.
\usepackage{algorithm}

\usepackage{shadow}
\usepackage{fancybox}
\usepackage{fancyhdr}

\usepackage{color}
\usepackage[usenames,dvipsnames,svgnames,table]{xcolor}

\definecolor{mypurple}{rgb}{.4,.0,.5}
\usepackage{xcolor}

\usepackage{color}

\definecolor{darkgreen}{rgb}{0, 0.4,0}
\newcommand{\dgr}[1]{\textcolor{darkgreen}{#1}}

\definecolor{purplebrown}{rgb}{0.5,0.1,0.6}

\definecolor{ultclupcol}{rgb}{0.1,0.5,0.5}

\definecolor{ultclupcola}{rgb}{.5,0,.5}
\newcommand{\ultclupcola}[1]{\textcolor{ultclupcola}{#1}}

\newcommand{\bl}[1]{\textcolor{blue}{#1}}

\definecolor{shadebrown}{rgb}{0.1,0.1,0.9}
\definecolor{lightblue}{rgb}{0.2,0,1}

\usepackage{fancybox}
\usepackage{graphicx}
\usepackage{epstopdf}
\usepackage{epsfig}
\usepackage{wrapfig}
\usepackage{subfigure}

\usepackage{xcolor}

\usepackage{tcolorbox}
\tcbuselibrary{skins}

\newtcbox{\xmybox}{on line,
arc=7pt,
before upper={\rule[-3pt]{0pt}{10pt}},boxrule=0pt,
boxsep=0pt,left=6pt,right=6pt,top=0pt,bottom=0pt,enhanced, coltext=blue, colback=white!10!yellow}

\newtcbox{\xmyboxa}{on line,
arc=7pt,
before upper={\rule[-3pt]{0pt}{10pt}},boxrule=0pt,
boxsep=0pt,left=6pt,right=6pt,top=0pt,bottom=0pt,enhanced, colback=white!10!yellow}

\newtcbox{\xmyboxb}{on line,
arc=7pt,
before upper={\rule[-3pt]{0pt}{10pt}},boxrule=1pt,colframe=darkgreen!100!blue,
boxsep=0pt,left=6pt,right=6pt,top=0pt,bottom=0pt,enhanced, colback=white!10!yellow}

\newtcbox{\xmyboxc}{on line,
arc=7pt,
before upper={\rule[-3pt]{0pt}{10pt}},boxrule=.7pt,colframe=blue!100!blue,
boxsep=0pt,left=6pt,right=6pt,top=0pt,bottom=0pt,enhanced, coltext=blue, colback=white!10!yellow}

\newtcbox{\xmytboxa}{on line,
arc=7pt,
before upper={\rule[-3pt]{0pt}{10pt}},boxrule=.0pt,colframe=pink!50!yellow,
boxsep=0pt,left=6pt,right=6pt,top=0pt,bottom=0pt,enhanced, coltext=white, colback=blue!40!red}

\newtcbox{\xmytboxb}{on line,
arc=7pt,
before upper={\rule[-3pt]{0pt}{10pt}},boxrule=.0pt,colframe=pink!50!yellow,
boxsep=0pt,left=6pt,right=6pt,top=0pt,bottom=0pt,enhanced, coltext=white, colback=white!40!green}

\usepackage[hyphens]{url}

\usepackage[colorlinks=true,
            linkcolor=black,
            urlcolor=blue,
            citecolor=purple]{hyperref}

\usepackage{breakurl}
\def\y{{\bf y}}
\def\v{{\bf v}}
\def\x{{\bf x}}

\def\x{{\mathbf x}}

\def\v{{\bf v}}
\def\x{{\bf x}}
\def\y{{\bf y}}

\def\h{{\bf h}}

\def\be{\begin{equation}}
\def\ee{\end{equation}}
\def\ba{\left[\begin{array}}
\def\ea{\end{array}\right]}

\def\v{{\bf v}}
\def\x{{\bf x}}
\def\y{{\bf y}}

\def\1{{\bf 1}}

\def\g{{\bf g}}
\def\0{{\bf 0}}

\def\erfinv{\mbox{erfinv}}
\def\erf{\mbox{erf}}
\def\erfc{\mbox{erfc}}

\def\mR{{\mathbb R}}

\def\mE{{\mathbb E}}

\newtheorem{theorem}{Theorem}

\begin{document}

\begin{singlespace}

%\title {Sparse ML linear regression -- CLuP's performance   %A tight variant of Gordon's escape through a mesh theorem
%%\footnote{ This work was supported in
%%part.}
%}

\title {Sparse linear regression -- CLuP achieves the ideal \emph{exact} ML   %A tight variant of Gordon's escape through a mesh theorem
%\footnote{ This work was supported in
%part.}
}
\author{
\textsc{Mihailo Stojnic
\footnote{e-mail: {\tt flatoyer@gmail.com}} }}
\date{}
\maketitle

\centerline{{\bf Abstract}} \vspace*{0.1in}

In this paper we revisit one of the classical statistical problems, the so-called sparse maximum-likelihood (ML) linear regression. As a way of attacking this type of regression, we present a novel CLuP mechanism that to a degree relies on the \bl{\textbf{Random Duality Theory (RDT)}} based algorithmic machinery that we recently introduced in \cite{Stojnicclupint19,Stojnicclupcmpl19,Stojnicclupplt19,Stojniccluplargesc20,Stojniccluprephased20}. After the initial success that the CLuP exhibited in achieving the exact ML performance while maintaining excellent computational complexity related properties in MIMO ML detection in \cite{Stojnicclupint19,Stojnicclupcmpl19,Stojnicclupplt19}, one would naturally expect that a similar type of success can be achieved in other ML considerations. The results that we present here confirm that such an expectation is indeed reasonable. In particular, within the sparse regression context, the introduced CLuP mechanism indeed turns out to be able to \bl{\textbf{\emph{achieve the ideal ML performance}}}. Moreover, it can substantially outperform some of the most prominent earlier state of the art algorithmic concepts, among them even the variants of the famous LASSO and SOCP from \cite{StojnicPrDepSocp10,StojnicGenLasso10,StojnicGenSocp10}. Also, our recent results presented in \cite{Stojniccluplargesc20,Stojniccluprephased20} showed that the CLuP has excellent \bl{\textbf{\emph{large-scale}}} and the so-called \bl{\textbf{\emph{rephasing}}} abilities. Since such large-scale algorithmic features are possibly even more desirable within the sparse regression context we here also demonstrate that the basic CLuP ideas can be reformulated to enable solving with a relative ease the regression problems with \bl{\textbf{\emph{several thousands}}} of unknowns. In addition to providing the fundamental theoretical concepts on which the desired CLuP relies, we also present a large set of results obtained through numerical experiments. As in our earlier CLuP related considerations, an excellent agreement between what the theory predicts and what one gets through the numerical experiments is observed as well.

\vspace*{0.25in} \noindent {\bf Index Terms: Sparse regression; Large scale CLuP; ML - estimation; Algorithms; Random duality theory}.

\end{singlespace}

%%%%%%%%%%%%%%%%%%%%%%%%%%%%%%%%%%%%%%%%%%%%%%%%%%%%%%%%%%%%%%%%%
\section{Introduction}
\label{sec:back}
%%%%%%%%%%%%%%%%%%%%%%%%%%%%%%%%%%%%%%%%%%%%%%%%%%%%%%%%%%%%%%%%%

Since our main interest in this paper will be a regression type of analysis, we start things off with recalling on the description of the standard regression models. Of central interest in such models is the following mathematical structure
\begin{eqnarray}\label{eq:linsys1a}
\y_i=f_i(A_{i,:}\x_{sol})+\sigma\v_i,
\end{eqnarray}
where $\y\in\mR^m$ is the vector of dependent variables, $A\in\mR^{m\times n}$ is the matrix whose rows ($A_{i,:},1\leq i\leq m$) are the vectors of the so-called independent variables, $\x_{sol}$ is the vector of unknown parameters, and $\v$ is the $\sigma$ scaled additive noise vector. Not that much of regression experience is fairly sufficient to recognize right here at the beginning that we use a substantially different notation from the one typically used in regression types of considerations. Namely, instead of (\ref{eq:linsys1a}), a typical regression consideration would be
\begin{eqnarray}\label{eq:linsys1b}
\y_i=f_i(X_{i,:}\beta)+\epsilon_i.
\end{eqnarray}
The reason for our deviation from the standard statistical routine will become much clearer as the presentation progresses. Here, we would just like to emphasize that, in order to make the presentation easier to follow, we find it way more convenient to utilize (\ref{eq:linsys1a}). Many of the concepts that we will present below will be related to those that we have already presented in \cite{Stojnicclupint19,Stojnicclupcmpl19,Stojnicclupplt19,Stojniccluplargesc20,Stojniccluprephased20} and earlier in \cite{StojnicCSetam09,StojnicCSetamBlock09,StojnicISIT2010binary,StojnicDiscPercp13,StojnicUpper10,StojnicGenLasso10,StojnicGenSocp10,StojnicPrDepSocp10,StojnicRegRndDlt10,Stojnicbinary16fin,Stojnicbinary16asym}. Since foundations of such concepts are discussed in great details in these papers we will here (instead of repeating some of these details) try to focus on presenting the most important differences. To achieve that we found that utilizing the notation similar to the one from \cite{Stojnicclupint19,Stojnicclupcmpl19,Stojnicclupplt19,Stojniccluplargesc20,Stojniccluprephased20} can be very helpful. In particular, simplification of the flow of this paper's presentation and establishing the connection to our earlier results both seemed easier to achieve while maintaining as much parallelism as possible with \cite{Stojnicclupint19,Stojnicclupcmpl19,Stojnicclupplt19,Stojniccluplargesc20,Stojniccluprephased20}.

Going back to (\ref{eq:linsys1a}) (and along the same lines of what we have just discussed), since our primary interest in this paper will be the so-called linear regression we will also assume $f_i(A_{i,:}\x_{sol})=A_{i,:}\x_{sol}$ (other types of regression we will discuss in separate papers). After such an assumption one arrives to the following linear noise corrupted model
\begin{eqnarray}\label{eq:linsys1}
\y=A\x_{sol}+\sigma\v.
\end{eqnarray}
We should also point out another rather small difference between (\ref{eq:linsys1}) and (\ref{eq:linsys1b}). Namely, in (\ref{eq:linsys1}) we also introduced a scaling factor $\sigma$ which will be useful in defining the so-called signal-to-noise (SNR) ratio -- a quantity of fundamental importance in studying and understanding both, the regression models that we consider here and the similar ones that we considered in \cite{Stojnicclupint19,Stojnicclupcmpl19,Stojnicclupplt19,Stojniccluplargesc20,Stojniccluprephased20}. For mathematical concreteness, we should also add that throughout the paper we will assume the so-called linear dimensions regime, i.e. we will assume that $m=\alpha n$, with $\alpha>0$ remaining a fixed constant of proportionality as $m$ and $n$ grow large.

Of course, as mentioned on many occasions in \cite{Stojnicclupint19,Stojnicclupcmpl19,Stojnicclupplt19,Stojniccluplargesc20,Stojniccluprephased20}, the above linear noise corrupted model is one of the key models not only in the linear regression but also in many other scientific fields, with information theory, signal processing, and machine learning probably being the most prominent ones (in some of these fields it is often referred to as the MIMO (multiple input -- multiple output) model). While the interest and the type of use of these models varies from one field to another they all have one thing in common. Namely, no matter what is the ultimate use of the model, the primary goal is almost always the estimation/recovery of the parameter vector $\x$. Such a question has been present in scientific literature for over two centuries and dates back to at least the second half of the eighteenth and early years of the nineteenth century and the work of Laplace, Boskovic, Legendre, and Gauss \cite{Legendre1805}. In the context of the linear regression the question is often a bit easier when compared to the corresponding ones in some other areas. Namely, within the linear regression setup, the usual structure of the estimation problem is such that the elements of matrix $A$ (its rows as vectors of independent variables) and the elements of $\y$ (dependent variables) are known whereas the noise vector $\v$ is not known. Even in such a scenario many criteria can be used for the recovery of $\x_{sol}$. Here we will rely on the so-called ML-criteria or the standard least squares which assumes that $\hat{\x}$ (as the estimate of $\x_{sol}$) is obtained as the solution to the following optimization problem
\begin{eqnarray}\label{eq:ml1}
\hat{\x}=\min_{\x\in{\cal X}}\|\y-A\x\|_2,
\end{eqnarray}
where ${\cal X}$ is the set of all allowed vectors $\x_{sol}$. It is rather clear that the estimation (\ref{eq:ml1}) is always of the least squares type. If one assumes that the elements of $\v$ are i.i.d. standard normals then it is also the so-called ML estimate. For the concreteness, we will in the rest of the paper rely on two additional assumption: 1) the elements of $A$ are also i.i.d. standard normals independent of the elements of $\v$; and 2) we will assume that the set of all permissible $\x_{sol}$, ${\cal X}$, contains all $k$-sparse vectors from $\mR^n$, with $k=\beta n$ and $\beta>0$, similarly to $\alpha$, remaining a fixed constant as $m$ and $n$ grow. The first of these two assumptions is just for the technical purposes and to basically ensure the easiness of the presentation (with fairly mild moment restrictions, any other distribution for $A$ in place of Gaussian works as well; more on this phenomenon can be found on various occasions in \cite{Stojnicclupint19,Stojnicclupcmpl19,Stojnicclupplt19,Stojniccluplargesc20,Stojniccluprephased20,StojnicCSetam09,StojnicCSetamBlock09,StojnicISIT2010binary,StojnicDiscPercp13,StojnicUpper10,StojnicGenLasso10,StojnicGenSocp10,StojnicPrDepSocp10,StojnicRegRndDlt10,Stojnicbinary16fin,Stojnicbinary16asym}). On the other hand the second assumption creates the setup of the so-called \bl{\textbf{\emph{sparse}}} ML linear regression. Such a setup has its own particularities and can be utilized to further modify/adapt the above ML recovery optimization problem (\ref{eq:ml1}).

%%%%%%%%%%%%%%%%%%%%%%%%%%%%%%%%%%%%%%%%%%%%%%%%%%%%%%%%%%%%%%%%%
\subsection{Sparse ML regression}
\label{sec:sparseml}
%%%%%%%%%%%%%%%%%%%%%%%%%%%%%%%%%%%%%%%%%%%%%%%%%%%%%%%%%%%%%%%%%

Since the sparse ML linear regression is also a well known research problem there has been quite a lot of work related to it over last several decades across many different fields. We leave a thorough discussion regarding all the prior work to survey papers, and instead mention only those works that are most directly relevant to what we present below. One of the most popular ways to incorporate the \emph{a priori} known fact that the unknown vector of parameters $\x$ is sparse is through the following so-called LASSO (least absolute shrinkage and selection operator) quadratic programming mechanism \cite{SanWill86,Tibsh96,CheDon95}
\begin{eqnarray}\label{eq:ml1lasso}
\hat{\x}=\min_{\x}(\|\y-A\x\|_2+c_{\ell_1}\|\x\|_1),
\end{eqnarray}
where $c_{\ell_1}$ is an appropriately chosen constant. The quality of the estimate $\hat{\x}$ and ultimately the success of the LASSO heavily relies on the choice of $c_{\ell_1}$. Typical heuristic based approaches suggest $c_{\ell_1}$ choices obtained based on the knowledge (a priori available or acquired through an estimation) that one has about the key structural components of the problem, the matrix $A$, and the vectors $\x$, $\y$, and $\v$. A lot of nice work regarding these and many other important aspects of the LASSO problems has been done over the last two decades (see, e.g. \cite{BunTsyWeg07,vandeGeer08,MeinYu09,JamRadLv09,W}). Moreover, over roughly the same period of time a lot of progress has been made researching other avenues to attack the sparse ML regression. In particular, a lot of progress has been made within the so-called noisy compressed sensing context where one of the most important underlying problems is precisely the problem from (\ref{eq:ml1}). Various alternatives to the LASSO have been proposed: the SOCP based basis pursuit \cite{CRT}, the Dantzig selector \cite{CanTao07,EfrHatTib07,CaiLv07}, and the LARS \cite{EfrHasJohnTib04}, just to name a few.

We will here though single out two lines of work that are probably the most relevant from the mathematical rigorousness point of view. In the context of interest here, the work of \cite{BayMon10,BayMon10lasso} and our own work \cite{StojnicGenLasso10,StojnicGenSocp10,StojnicPrDepSocp10} were the first to suggest, mathematically rigorously justifiable ways to utilize the knowledge about the structure of the above LASSO problem. \cite{BayMon10,BayMon10lasso} had done so through a state evolution type of analysis of a noise adapted variant of the approximate message passing (AMP) algorithm that was earlier introduced in the corresponding noiseless scenario within the above mentioned compressed sensing context (see, e.g. \cite{DonMalMon09,DonMalMon10}). On the other hand, in \cite{StojnicGenLasso10,StojnicGenSocp10,StojnicPrDepSocp10} we relied on the fundamentals of our \bl{\textbf{Random Duality Theory (RDT)}} and developed a framework to directly analyze the above LASSO algorithm. Through such an analysis we obtained an exact characterization of all parameters relevant in studying the performance of (\ref{eq:ml1lasso}). Needless to say, among them was also the exact value of the above mentioned so-called minimum MSE (mean square error) achieving optimal choice for $c_{\ell_1}$.

\begin{figure}[htb]
%\begin{minipage}[b]{.5\linewidth}
\centering
\centerline{\epsfig{figure=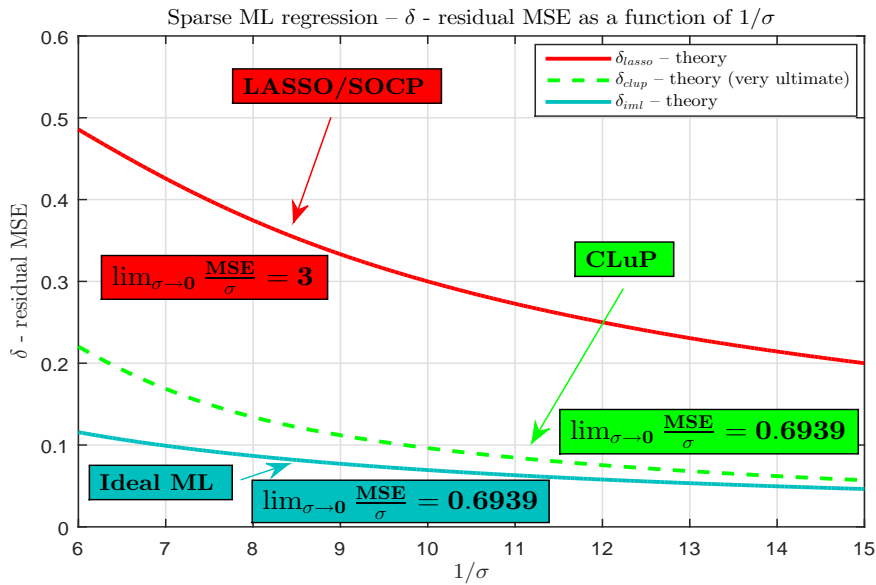,width=13.5cm,height=8cm}}
%\end{minipage}
%\begin{minipage}[b]{.5\linewidth}
%\centering
%\centerline{\epsfig{figure=finprerral08.eps,width=9cm,height=6.5cm}}
%\end{minipage}
\caption{Comparison of $\delta$ (MSE (mean square error)) as a function of $1/\sigma$; $\alpha=0.5$; $\beta=0.1625$ ($n\rightarrow \infty$)}
\label{fig:HLmmselassoclupverybest}
\end{figure}

In this paper, we move things much further and design a novel CLuP (controlled loosening-up) type of algorithm that massively outperforms the LASSO variants from \cite{StojnicGenLasso10,StojnicGenSocp10,StojnicPrDepSocp10}. Moreover, not only are the residual MSEs often even three times smaller than the corresponding state of the art LASSO ones, they can be achieved though large-scale implementations that allow us to solve problems with several thousands of unknowns rather quickly with theoretically minimal \textbf{\emph{quadratic}} complexity per iteration (only a single matrix-vector multiplication is the core computation within each of the algorithm's iterations). In Figure \ref{fig:HLmmselassoclupverybest} we provide a brief preview as to what kind of performance one might ultimately expect from the CLuP. While we leave a thorough discussion regarding all the details related to the plots shown in the figure for later sections of the paper, here at the beginning we just briefly mention that the CLuP method does seem very powerful and in addition to being capable of handling very large scale problems it also provides an avenue towards achieving even the ideal ML performance (unbreakable within the sparse ML context).

We will organize the paper in the following way. We will first introduce the main CLuP mechanism that can handle the ML sparse regression. While doing so we will also recall on the fundamentals behind the CLuP concepts from \cite{Stojnicclupint19,Stojnicclupcmpl19,Stojnicclupplt19,Stojniccluplargesc20,Stojniccluprephased20}. The corresponding large-scale implementation will be discussed as well. We will then present both, the theoretical estimates as well as the results that can be obtained by practically running the designed CLuP algorithms. At the and we will provide a summarized discussion regarding the entire presentation with quite a few conclusions and avenues for future work.

%%%%%%%%%%%%%%%%%%%%%%%%%%%%%%%%%%%%%%%%%%%%%%%%%%%%%%%%%%%%%%%%%
\section{\bl{CLuP} -- sparse ML regression}
\label{sec:clupfund}
%%%%%%%%%%%%%%%%%%%%%%%%%%%%%%%%%%%%%%%%%%%%%%%%%%%%%%%%%%%%%%%%%

In this section we will present the basic CLuP mechanism that can be used to attack the sparse ML regression. Before doing so, we will need to recall on a few results that we have created in some of our earlier works. Along the same lines, we do mention that throughout the exposition we will assume a decent level of familiarity with some of the key concepts presented in \cite{Stojnicclupint19,Stojnicclupcmpl19,Stojnicclupplt19,Stojniccluplargesc20,Stojniccluprephased20,StojnicCSetam09,StojnicCSetamBlock09,StojnicISIT2010binary,StojnicDiscPercp13,StojnicUpper10,StojnicGenLasso10,StojnicGenSocp10,StojnicPrDepSocp10,StojnicRegRndDlt10,Stojnicbinary16fin,Stojnicbinary16asym}. As the context dictates and to ensure the smoothness of the presentation and the easiness of the following, we will on occasion try to briefly reemphasize some of the key ideas from these papers; on the other hand, for the full detailed story we will usually refer the interested reader to \cite{Stojnicclupint19,Stojnicclupcmpl19,Stojnicclupplt19,Stojniccluplargesc20,Stojniccluprephased20,StojnicCSetam09,StojnicCSetamBlock09,StojnicISIT2010binary,StojnicDiscPercp13,StojnicUpper10,StojnicGenLasso10,StojnicGenSocp10,StojnicPrDepSocp10,StojnicRegRndDlt10,Stojnicbinary16fin,Stojnicbinary16asym}.
Similarly to what was the case in our earlier sparse ML regression considerations in \cite{StojnicGenLasso10,StojnicGenSocp10,StojnicPrDepSocp10} we will often find it useful to work with alternative versions of the LASSO problem from (\ref{eq:ml1lasso}). Namely, the following so-called SOCP variant of the LASSO from (\ref{eq:ml1lasso}) was shown in \cite{StojnicGenLasso10,StojnicGenSocp10,StojnicPrDepSocp10} to be as powerful as the LASSO itself
\begin{eqnarray}
\hat{\x}=\mbox{arg}\min_{\x} & & \|\x\|_1  \nonumber \\
\mbox{subject to} & & \|\y-A\x\|_2\leq r_{socp}. \label{eq:ml1socp}
\end{eqnarray}
Assuming that $\hat{\x}_{lasso}$ is the solution of (\ref{eq:ml1lasso}) and that $\hat{\x}_{socp}$ is the solution of (\ref{eq:ml1socp}), one of the main takeaways from \cite{StojnicGenLasso10,StojnicGenSocp10,StojnicPrDepSocp10} was that a careful choice for $c_{\ell_1}$ and $r_{socp}$ can ensure that the corresponding LASSO and SOCP residual averaged MSEs are equal to each other. In other words one has
\begin{eqnarray}
\delta_{lasso} & \triangleq & \mE\|\hat{\x}_{lasso}-\x_{sol}\|_2\nonumber \\
\delta_{socp} & \triangleq & \mE\|\hat{\x}_{socp}-\x_{sol}\|_2, \label{eq:deltaequiv1}
\end{eqnarray}
and
\begin{equation}
\delta_{lasso} =\delta_{socp}. \label{eq:deltaequiv2}
\end{equation}
We do mention right here at the beginning that in the rest of the paper whenever we discuss a quantity that has the concentrating property its averaged value will often be assumed as its real value due to assumed large dimensional settings and overwhelming concentrating probabilities (as discussed on quite a few occasions throughout \cite{Stojnicclupint19,Stojnicclupcmpl19,Stojnicclupplt19,Stojniccluplargesc20,Stojniccluprephased20,StojnicCSetam09,StojnicCSetamBlock09,StojnicISIT2010binary,StojnicDiscPercp13,StojnicUpper10,StojnicGenLasso10,StojnicGenSocp10,StojnicPrDepSocp10,StojnicRegRndDlt10,Stojnicbinary16fin,Stojnicbinary16asym}, these concentrations are exponential in $n$). The following theorem is a bit more precise and generalized version of the above statements from (\ref{eq:deltaequiv1}) and (\ref{eq:deltaequiv2}).
\begin{theorem}(\bl{\textbf{Plain LASSO/SOCP}} -- equivalence \cite{StojnicGenLasso10,StojnicGenSocp10,StojnicPrDepSocp10})
  Let $\hat{\x}_{lasso}$ and $\hat{\x}_{socp}$ be the solutions to (\ref{eq:ml1lasso}) and (\ref{eq:ml1socp}), respectively. Also, assume the above mentioned statistical linear large dimensional scenario with $\alpha=\lim_{n \rightarrow \infty}\frac{m}{n}$ and $\beta=\lim_{n \rightarrow \infty}\frac{k}{n}$, and let $\alpha_w$ and $\beta$ satisfy the following \textbf{fundamental $\ell_1$ phase transition} characterization \cite{StojnicCSetam09,StojnicUpper10}
  \begin{equation}\label{eq:thm1eq1}
  \frac{(1-\beta)\exp(-(\erfinv(\frac{1-\alpha_w}{1-\beta}))^2)}{\sqrt{\pi}\alpha_w\erfinv(\frac{1-\alpha_w}{1-\beta})}=1.
  \end{equation}
  For $c_{\ell_1}=\sqrt{2}\erfinv(\frac{1-\alpha_w}{1-\beta})$ and $r_{socp}=\sigma\sqrt{(\alpha-\alpha_w)n}$ one has that $\forall \x_{sol}$
  \begin{eqnarray}\label{eq:thm1eq2}
  P\left ( \|\hat{\x}_{lasso}-\x_{sol}\|_2\leq \sigma\sqrt{\frac{\alpha_w}{\alpha-\alpha_w}}\right )&\rightarrow & 1\nonumber \\
  P\left ( \|\hat{\x}_{socp}-\x_{sol}\|_2\leq \sigma\sqrt{\frac{\alpha_w}{\alpha-\alpha_w}}\right )&\rightarrow & 1.
  \end{eqnarray}
  Moreover, $\exists \x_{sol}$ such that
  \begin{eqnarray}\label{eq:thm1eq2}
  P\left ( \|\hat{\x}_{lasso}-\x_{sol}\|_2\rightarrow \sigma\sqrt{\frac{\alpha_w}{\alpha-\alpha_w}}\right )&\rightarrow & 1\nonumber \\
  P\left ( \|\hat{\x}_{socp}-\x_{sol}\|_2\rightarrow \sigma\sqrt{\frac{\alpha_w}{\alpha-\alpha_w}}\right )&\rightarrow & 1.
  \end{eqnarray}\label{thm:lassosocp}
\end{theorem}
The above theorem is of the so-called worst-case achieving type. Namely, it characterizes the performance of the two underlying algorithms through the behavior of the so-called worst-case MSE. However, it also ensures that such a worst case behavior can indeed be achieved. In fact, detailed analyses in \cite{StojnicGenLasso10,StojnicGenSocp10,StojnicPrDepSocp10} established that the worst-case MSE can be achieved for $k$-sparse $\x_{sol}$ that have equal magnitudes of the nonzero components that tend to infinity (\cite{StojnicPrDepSocp10} also went a bit further and established corresponding results for any value of the nonzero magnitudes). To maintain parallelism with some of the results presented in \cite{StojnicGenLasso10,StojnicGenSocp10,StojnicPrDepSocp10} we will below for the concreteness often assume this very same worst-case achieving scenario where all nonzero components of $\x_{sol}$ are equal to each other.

Of course, as mentioned above, one of the main points of \cite{StojnicGenLasso10,StojnicGenSocp10,StojnicPrDepSocp10} and the above theorem is that it succeeds in establishing that the SOCP from (\ref{eq:ml1socp}) is ultimately as powerful as the LASSO from (\ref{eq:ml1lasso}). On the other hand, we will below design an algorithm that substantially improves over both of them.

%%%%%%%%%%%%%%%%%%%%%%%%%%%%%%%%%%%%%%%%%%%%%%%%%%%%%%%%%%%%%%%%%
\subsection{Basic \bl{CLuP} for sparse regression}
\label{sec:clupspreg}
%%%%%%%%%%%%%%%%%%%%%%%%%%%%%%%%%%%%%%%%%%%%%%%%%%%%%%%%%%%%%%%%%

Our recent series of papers \cite{Stojnicclupint19,Stojnicclupcmpl19,Stojnicclupplt19} introduced the so-called Controlled loosening-up (\bl{\textbf{CLuP}}) algorithmic mechanism that turns out to be very powerful and useful in solving optimization problems that are in the classical complexity theory typically viewed as hard. To present the main ideas behind the entire CLuP concept and to ultimately demonstrate its practical power we in \cite{Stojnicclupint19,Stojnicclupcmpl19,Stojnicclupplt19} used the so-called MIMO ML detection as the benchmark problem. Here the underlying problem is also of the ML type. However, the structure of the unknown vector of parameters is substantially different from the one considered in \cite{Stojnicclupint19,Stojnicclupcmpl19,Stojnicclupplt19} and one has to be a bit more careful when designing the corresponding CLuP mechanism.

Given that this time ${\cal X}$ is the set of all $k$-sparse vectors in $\mR^n$, we consider the following CLuP type iterative procedure. Assume that $\x^{(0)}$ is a randomly chosen binary vector from $\mR^n$ ($\x^{(0)}$ can (but it doesn't even have to) be from ${\cal X}$). We will think of such a vector as the starting estimate for $\x_{sol}$. Then one continues building further estimates $\x^{(i)},i>0$ through the following mechanism
\begin{eqnarray}
\x^{(i+1)}=\frac{\x^{(i+1,s)}}{\|\x^{(i+1,s)}\|_2} \quad \mbox{with}\quad \x^{(i+1,s)}=\mbox{arg}\min_{\x} & & -(\x^{(i)})^T\x+c_{\ell_1}\|\x\|_1  \nonumber \\
\mbox{subject to} & & \|\y-A\x\|_2\leq r. \label{eq:clup1}
\end{eqnarray}
The above procedure is of course very simple and in way similar to some of our earlier CLuP considerations \cite{Stojnicclupint19,Stojnicclupcmpl19,Stojnicclupplt19}. However, there are a couple of key differences that make eventual utilization of the above mechanism substantially different from the ones discussed in \cite{Stojnicclupint19,Stojnicclupcmpl19,Stojnicclupplt19}. Namely, with a little bit of familiarity with \cite{Stojnicclupint19,Stojnicclupcmpl19,Stojnicclupplt19} one recognizes that $r$ is again the so-called radius that will play one of the most important roles in the overall success of the algorithm. Following the trend set in \cite{Stojnicclupint19,Stojnicclupcmpl19,Stojnicclupplt19} as well as in \cite{Stojniccluprephased20,Stojniccluplargesc20} we will think of $r$ as being a multiple of $r_{socp}$, i.e. $r=r_{sc}r_{socp}$, where $r_{socp}$ in a way corresponds to the $r_{socp}$ in (\ref{eq:ml1socp}) (moreover, we will quite often throughout the paper assume that $r_{socp}$ is as in Theorem \ref{thm:lassosocp}). Now, looking carefully at (\ref{eq:clup1}) and comparing it to the similar CLuP foundations presented in \cite{Stojnicclupint19,Stojnicclupcmpl19,Stojnicclupplt19} one observes the appearance of a new term $c_{\ell_1}\|\x\|_1$. This term in a way emulates the corresponding one in (\ref{eq:ml1lasso}). In a more informal language, as in (\ref{eq:ml1lasso}), its role here is to basically through a $c_{\ell_1}$ scaling emphasize the sparse structure of $\x_{sol}$.

%%%%%%%%%%%%%%%%%%%%%%%%%%%%%%%%%%%%%%%%%%%%%%%%%%%%%%%%%%%%%%%%%
\subsection{\bl{CLuP}'s large-scale implementation}
\label{sec:cluplargescimpl}
%%%%%%%%%%%%%%%%%%%%%%%%%%%%%%%%%%%%%%%%%%%%%%%%%%%%%%%%%%%%%%%%%

At this point though there is no real reason to believe that the above mechanism would be of any use, let alone that it can be substantially better than the state of the art LASSOs and SOCPs discussed earlier. In fact, although the procedure looks of a similar level of simplicity as the CLuPs from \cite{Stojnicclupint19,Stojnicclupcmpl19,Stojnicclupplt19}, things are not necessarily super simple either. Basically, in addition to the controlling radius $r$ (or its a scaled variant $r_{sc}$) one now has to be careful with respect to the choice of $c_{\ell_1}$ as well. Still, if one is somehow by a miracle super lucky that all these things can be handled and that their handling is so successful that the resulting procedure is indeed soundly better than the LASSO or SOCP, one then may expect that further analogies with some of the CLuP concepts can be established as well. That in the first place relates to the CLuP's large-scale abilities. Namely, as discussed in \cite{Stojniccluprephased20,Stojniccluplargesc20}, not only are the CLuPs from \cite{Stojnicclupint19,Stojnicclupcmpl19,Stojnicclupplt19} substantially improving on the existing convexity based state of the art algorithms, they are doing so while allowing for large-scale implementations that eventually can handle with an ease problems with several thousands of unknowns.

Thinking further along the lines of \cite{Stojniccluprephased20,Stojniccluplargesc20} one can try to establish a corresponding large-scale CLuP implementation that would be useful for the problems of interest here. For example, one can first transform the above basic CLuP mechanism into the following
\begin{eqnarray}
\xi_p\triangleq\lim_{n\rightarrow \infty}\mE\min_{\x} & & \xi_{LS}  \nonumber \\
\mbox{subject to} & & \xi_{LS}=-\|\x\|_2+\hat{c}_{\ell_1}\|\x\|_1+\hat{\gamma}_1(\|\y-A\x\|_2- r), \label{eq:LSclup1}
\end{eqnarray}
where hopefully  $\hat{c}_{\ell_1}$ and $\hat{\gamma}_1$ can be obtained through a machinery similar to the ones presented in \cite{Stojniccluprephased20,Stojniccluplargesc20}. Moreover, following further \cite{Stojniccluplargesc20} one also has
\begin{equation}
\frac{d\xi_{LS}}{d\x} =-\frac{\x}{\|\x\|_2}+\hat{c}_{\ell_1}\mbox{sign}(\x)+\hat{\gamma}_1\frac{-A^T(\y-A\x)}{\|\y-A\x\|_2}, \label{eq:LSclup2}
\end{equation}
and after equalling the derivative with zero
\begin{equation}
\frac{d\xi_{LS}}{d\x}=0 \Longleftrightarrow -\x\|\y-A\x\|_2+\hat{c}_{\ell_1}\mbox{sign}(\x)\|\x\|_2\|\y-A\x\|_2-\hat{\gamma}_1A^T\y\|\x\|_2+\hat{\gamma}_1A^TA\x\|\x\|_2=0. \label{eq:LSclup3}
\end{equation}
As discussed in \cite{Stojniccluplargesc20} the above equation can be solved in many different ways. Here, we follow into the footsteps of \cite{Stojniccluplargesc20} and choose probably one of the simplest possible methods, i.e. we choose a simple contraction with a scaled regularization. In other words,
\begin{eqnarray}
& &  -\x\|\y-A\x\|_2+\hat{c}_{\ell_1}\mbox{sign}(\x)\|\x\|_2\|\y-A\x\|_2-\hat{\gamma}_1A^T\y\|\x\|_2+\hat{\gamma}_1A^TA\x\|\x\|_2=0 \nonumber \\
&\Longleftrightarrow & c_{q,2}\x-\x\|\y-A\x\|_2-c_{q,2}\x+\hat{c}_{\ell_1}\mbox{sign}(\x)\|\x\|_2\|\y-A\x\|_2-\hat{\gamma}_1A^T\y\|\x\|_2+\hat{\gamma}_1A^TA\x\|\x\|_2=0\nonumber \\
&\Longleftrightarrow & \x(c_{q,2}-\|\y-A\x\|_2)=c_{q,2}\x-\hat{c}_{\ell_1}\mbox{sign}(\x)\|\x\|_2\|\y-A\x\|_2+\hat{\gamma}_1A^T\y\|\x\|_2-\hat{\gamma}_1A^TA\x\|\x\|_2 \nonumber \\
&\Longleftrightarrow & \x=\frac{c_{q,2}\x-\hat{c}_{\ell_1}\mbox{sign}(\x)\|\x\|_2\|\y-A\x\|_2+\hat{\gamma}_1A^T\y\|\x\|_2-\hat{\gamma}_1A^TA\x\|\x\|_2}{c_{q,2}-\|\y-A\x\|_2}, \label{eq:LSclup4}
\end{eqnarray}
where, of course, $c_{q,2}$ is an appropriately selected constant. This is then enough to establish the above mentioned contraction
\begin{equation}
\x^{(i+1)} = \frac{c_{q,2}\x^{(i)}-\hat{c}_{\ell_1}\mbox{sign}(\x^{(i)})\|\x^{(i)}\|_2\|\y-A\x^{(i)}\|_2+\hat{\gamma}_1A^T\y\|\x^{(i)}\|_2-\hat{\gamma}_1A^TA\x^{(i)}\|\x^{(i)}\|_2}{c_{q,2}-\|\y-A\x^{(i)}\|_2}.\label{eq:LSclup5}
\end{equation}
Moreover, assuming that in the limit $\|\x^{(i)}\|_2\rightarrow \sqrt{\hat{c}_2}$ and $\|\y-A\x^{(i)}\|_2\rightarrow r=r_{sc}r_{socp}$, one also has
\begin{equation}
\x^{(i+1)}=\frac{c_{q,2}\x^{(i)}-\hat{c}_{\ell_1}\mbox{sign}(\x^{(i)})\sqrt{\hat{c}_2}r_{sc}r_{socp}+\hat{\gamma}_1A^T\y\sqrt{\hat{c}_2}-\hat{\gamma}_1A^TA\x^{(i)}\sqrt{\hat{c}_2}}{c_{q,2}-r_{sc}r_{socp}}.\label{eq:LSclup6}
\end{equation}
It is not that difficult to see that the above contraction mechanism inherits all the excellent features that its predecessors from \cite{Stojniccluprephased20,Stojniccluplargesc20} posses. The complexity per iteration is theoretically minimal, i.e. it is quadratic and it includes only a single matrix-vector multiplication which amounts to $mn$ basic addition/multiplication operations. All other discussions from \cite{Stojniccluprephased20,Stojniccluplargesc20} regarding practical running of this contraction (in particular, regarding rerunning, the choices of stopping criteria, the sufficient number of iterations and so on) remain in place. We skip all these details and instead refer the interested reader to \cite{Stojniccluprephased20,Stojniccluplargesc20} for a more thorough discussion in these directions.

%%%%%%%%%%%%%%%%%%%%%%%%%%%%%%%%%%%%%%%%%%%%%%%%%%%%%%%%%%%%%%%%%
\subsection{\bl{CLuP}'s analytical foundation}
\label{sec:clupanal}
%%%%%%%%%%%%%%%%%%%%%%%%%%%%%%%%%%%%%%%%%%%%%%%%%%%%%%%%%%%%%%%%%

Carefully looking at the above discussion, one can observe that there are several key things that would need to be addressed before one can actually utilize the above algorithmic structure. The one that comes to mind first is the choice of critical parameters, $r_{sc}$, $\hat{c}_2$ , $\hat{\gamma}_1$ and $\hat{c}_{\ell_1}$. A priori that is not an easy task and requires a very careful sequence of considerations.

As in \cite{Stojniccluprephased20,Stojniccluplargesc20}, we will rely on the algorithmic properties of the \bl{\textbf{Random Duality Theory}} and a large set of results that we created in \cite{StojnicCSetam09,StojnicCSetamBlock09,StojnicISIT2010binary,StojnicDiscPercp13,StojnicUpper10,StojnicGenLasso10,StojnicGenSocp10,StojnicPrDepSocp10,StojnicRegRndDlt10,Stojnicbinary16fin,Stojnicbinary16asym}. Also, to ensure the easiness of the following, we will try to parallel the relevant portions of the exposition from \cite{Stojniccluprephased20,Stojniccluplargesc20} as often as possible. However, we will skip many details that are similar to the corresponding ones from \cite{Stojniccluprephased20,Stojniccluplargesc20} and instead focus on emphasizing the key differences. We start by considering the optimization problem that one obtains at the end of the CLuP converging process defined in (\ref{eq:clup1})
\begin{eqnarray}
\min_{\x} & & -\|\x\|_2+c_{\ell_1}\|\x\|_1  \nonumber \\
\mbox{subject to} & & \|\y-A\x\|_2\leq r. \label{eq:clup2}
\end{eqnarray}
Plugging the $\y$ from (\ref{eq:linsys1}) into (\ref{eq:clup2}) we have
\begin{eqnarray}
\min_{\x} & & -\|\x\|_2+c_{\ell_1}\|\x\|_1  \nonumber \\
\mbox{subject to} & & \|[A \v]\begin{bmatrix}\x_{sol}-\x\\\sigma\end{bmatrix}\|_2\leq r. \label{eq:clup4}
\end{eqnarray}
 To utilize the \bl{\textbf{Random Duality Theory}} considerations from \cite{Stojnicclupint19,Stojnicclupcmpl19,Stojnicclupplt19}, we first define the two critical concentrating parameters, $c_1$ and $c_2$ as
\begin{eqnarray}
c_2 & = & \|\x\|_2^2\nonumber \\
c_1 & = & (\x_{sol})^T\x, \label{eq:clup3}
\end{eqnarray}
and then transform (\ref{eq:clup4}) into
\begin{eqnarray}
\min_{c_2\in[0,1]}\min_{\|\x\|_2^2=c_2} & & -\sqrt{c_2}+c_{\ell_1}\|\x\|_1  \nonumber \\
\mbox{subject to} & & \|[A \v]\begin{bmatrix}\x_{sol}-\x\\\sigma\end{bmatrix}\|_2\leq r. \label{eq:clup4ga}
\end{eqnarray}
Forming the standard Largangian the above optimization can be reformulated in the following way
\begin{eqnarray}
\min_{c_2\in[0,1]}\min_{\|\x\|_2^2=c_2} \max_{\gamma_1} & & -\sqrt{c_2}+c_{\ell_1}\|\x\|_1 +\gamma_1\left (\mbox{max}_{\|\lambda\|_2=1}\lambda^T\left ([A \v]\begin{bmatrix}\x_{sol}-\x\\\sigma\end{bmatrix}\right )- r\right ). \label{eq:clup5}
\end{eqnarray}
Moreover, assuming the concentration of $\gamma_1$ one also has
\begin{eqnarray}
\min_{c_2\in[0,1]}\max_{\gamma_1}\min_{\|\x\|_2^2=c_2} \max_{\|\lambda\|_2=1}  & & -\sqrt{c_2}+c_{\ell_1}\|\x\|_1 +\gamma_1\lambda^T\left ([A \v]\begin{bmatrix}\x_{sol}-\x\\\sigma\end{bmatrix}\right )- \gamma_1r. \label{eq:clup5a}
\end{eqnarray}
One can now proceed in the standard \bl{\textbf{RDT}} fashion outlined in \cite{Stojnicclupint19}.

\vspace{.1in}
\noindent \xmyboxc{\textbf{\emph{\dgr{Forming/handling the Random dual}}}}

\vspace{.1in}
 One of the key aspects of the \bl{\textbf{RDT}} is the introduction of the so-called \bl{\textbf{random dual}} to the above primal (see, e.g. \cite{StojnicCSetam09,StojnicISIT2010binary,StojnicDiscPercp13,StojnicGenLasso10,StojnicGenSocp10,StojnicPrDepSocp10,StojnicRegRndDlt10}). Following into the footsteps of a long line of our earlier work we then have for the \bl{\textbf{random dual}}
\begin{eqnarray}
\min_{c_2}\min_{c_1}\max_{\gamma_1,\nu}\min_{\|\x\|_2=c_2}\max_{\|\lambda\|_2=1} & &  \xi_{RD}^{(1)}, \label{eq:clup10}
\end{eqnarray}
with
\begin{equation}
\xi_{RD}^{(1)} =-\sqrt{c_2}+c_{\ell_1}\|\x\|_1+\gamma_1 (\lambda^T\g\sqrt{\|\x_{sol}-\x\|_2^2+\sigma^2}+\|\lambda\|_2(\h^T(\x_{sol}-\x)+h_0\sigma)-r) +\nu((\x_{sol})^T\x-c_1)
\label{eq:clup10a}
\end{equation}
and the components of $\g$ and $\h$ are $m$ and $n$ dimensional standard normal vectors with i.i.d. components ($h_0$ is also a standard normal random variable independent of all other random variables). After solving over $\lambda$ and ignoring $h_0$ due to the concentration the above can be transformed into
\begin{eqnarray}
\min_{c_2}\min_{c_1}\max_{\gamma_1,\gamma,\nu}\min_{\x} & & \xi_{RD}^{(2)}, \label{eq:clup10b}
\end{eqnarray}
with
\begin{equation}
\xi_{RD}^{(2)} =-\sqrt{c_2}+c_{\ell_1}\|\x\|_1+\gamma_1 (\|\g\|_2\sqrt{\|\x_{sol}-\x\|_2^2+\sigma^2}+\h^T(\x_{sol}-\x)-r) +\nu((\x_{sol})^T\x-c_1)+\gamma (\|\x\|_2^2-c_2). \label{eq:clup10c}
\end{equation}
We can then define $\xi_{RD,\gamma_1}(\alpha,\sigma;c_2,c_1,\gamma,\nu)$ (as the so-called optimizing objective of the random dual) in the following way
\begin{equation}
\xi_{RD,\gamma_1}(\alpha,\sigma;c_2,c_1,\gamma,\nu)=\lim_{n\rightarrow\infty}\frac{1}{\sqrt{n}}\mE\min_{\x} \xi_{RD}^{(2)}. \label{eq:clup10d}
\end{equation}
One can then proceed further as in \cite{StojnicCSetam09,StojnicISIT2010binary,StojnicDiscPercp13,StojnicGenLasso10,StojnicGenSocp10,StojnicPrDepSocp10,StojnicRegRndDlt10,Stojnicclupint19} and handle the above optimization. To that end we follow say \cite{StojnicDiscPercp13,Stojnicclupint19} and introduce
\begin{eqnarray}
f_{box}(\h;c_2,c_1)=\max_{\gamma_1,\gamma,\nu}\min_{\x} & & c_{\ell_1}\|\x\|_1-\gamma_1 \h^T\x +\nu((\x_{sol})^T\x-c_1)+\gamma (\|\x\|_2^2-c_2). \label{eq:clup12}
\end{eqnarray}
For the concreteness we will also assume the scenario that corresponds to the one from Theorem \ref{thm:lassosocp}, i.e. we will assume that all nonzero components of $\x_{sol}$ are equal to $1/\sqrt{k}$ (also for the concreteness and without a loss of generality, we will assume that they are the first $k$ components of $\x_{sol}$). We do however, emphasize that this fact is only for the concreteness of the analysis purposes. In other words, this is not a piece of knowledge that is available beforehand and that potentially can be utilized in the algorithm's design. Keeping this in mind and after a bit of rescaling one has similarly to (110) from \cite{StojnicDiscPercp13}
\begin{eqnarray}
f_{box}(\h;c_2,c_1)  =  \max_{\gamma,\nu} & & \frac{1}{\sqrt{n}}\left (\sum_{i=1}^{n}f_{box}^{(1)}(\h_i,\gamma_1,\gamma,\nu)\right )-\nu c_1\sqrt{\beta n}-\gamma c_2\sqrt{n},\label{eq:clup13}
\end{eqnarray}
where
\begin{equation}
f_{box}^{(1)}(\h_i,\gamma_1,\gamma,\nu)=\begin{cases}-\frac{(|\gamma_1\h_i-\nu|-c_{\ell_1})^2}{4\gamma}, & i\leq k\\
-\frac{(|\gamma_1\h_i|-c_{\ell_1})^2}{4\gamma}, & i>k,
\end{cases}\label{eq:clup14}
\end{equation}
and $\gamma$, $\nu$, and $c_{\ell_1}$ are $\sqrt{n}$, $\sqrt{k}$, and $\sqrt{n}$ respectively, scaled versions of the corresponding $\gamma$, $\nu$, and $c_{\ell_1}$ from (\ref{eq:clup12}) (here as well as on many other occasions we will use the same notation for dimension scaled quantities; this will typically be clear from the context and we will often skip emphasizing it). One also has that the optimizing $\x_i$ is given as
\begin{equation}
\x_i=\begin{cases}-\frac{|\gamma_1\h_i-\nu|-c_{\ell_1}}{2\gamma}, & i\leq k\\
-\frac{|\gamma_1\h_i|-c_{\ell_1}}{2\gamma}, & i>k
\end{cases}.\label{eq:clup14a}
\end{equation}
One then proceeds with solving the integrals and obtains
\begin{equation}
\mE f_{box}^{(1)}(\h_i,\gamma_1,\gamma,\nu)=\begin{cases}-\frac{I_{11}(\gamma_1,\nu,c_{\ell_1})+I_{12}(\gamma_1,\nu,c_{\ell_1})}{4\gamma}, & i\leq k\\
-\frac{I_{11}(\gamma_1,0,c_{\ell_1})+I_{12}(\gamma_1,0,c_{\ell_1})}{4\gamma}, & i>k
\end{cases},\label{eq:clup15}
\end{equation}
where
\begin{eqnarray}
I_{11}(\gamma_1,\nu,c_{\ell_1}) &  = &  (0.5\erfc((c_{\ell_1}-\nu)/\gamma_1/\sqrt{2}) (\gamma_1^2 + (\nu- c_{\ell_1})^2) - \frac{\gamma_1}{\sqrt{2\pi}} \exp(-((\nu-c_{\ell_1})/\gamma_1)^2/2) (c_{\ell_1}-\nu))\nonumber \\
I_{12}(\gamma_1,\nu,c_{\ell_1}) &  = &  (0.5\erfc((c_{\ell_1}+\nu)/\gamma_1/\sqrt{2}) (\gamma_1^2 + (\nu+ c_{\ell_1})^2) - \frac{\gamma_1}{\sqrt{2\pi}} \exp(-((\nu+c_{\ell_1})/\gamma_1)^2/2) (c_{\ell_1}+\nu)).\nonumber \\ \label{eq:clup16}
\end{eqnarray}
Combining (\ref{eq:clup10})-(\ref{eq:clup16}) one finally has for $\xi_{RD}(\alpha,\sigma;c_2,c_1,\gamma,\nu)$
\begin{equation}
\xi_{RD}(\alpha,\sigma;c_2,c_1,\gamma,\nu)=-\sqrt{c_2}+\gamma_1\sqrt{\alpha}\sqrt{1-2c_1+c_2+\sigma^2}
-\frac{I}{4\gamma}-\gamma_1 r-\nu c_1\sqrt{\beta}-\gamma c_2, \label{eq:clup17}
\end{equation}
 with
\begin{equation}
I= \beta(I_{11}(\gamma_1,\nu,c_{\ell_1})+I_{12}(\gamma_1,\nu,c_{\ell_1}))+(1-\beta)(I_{11}(\gamma_1,0,c_{\ell_1})+I_{12}(\gamma_1,0,c_{\ell_1})),\label{eq:clup17a}
\end{equation}
and $r$ being the $\sqrt{n}$ scaled version of $r$ from (\ref{eq:clup10c}). Following into the footsteps of \cite{Stojniccluplargesc20} one can utilize \cite{Stojnicclupint19}'s Theorem 1, and establish the following optimization problem as the main object of interest
\begin{equation}
\min_{c_2\in[0,1]}\min_{c_1\in[0,\sqrt{c_2}]}\max_{\gamma_1,\gamma,\nu}  \xi_{RD,\gamma_1}(\alpha,\sigma;c_2,c_1,\gamma,\nu). \label{eq:clupg5a}
\end{equation}
After quickly solving the above optimization over $\gamma$ one also has
\begin{equation}
\min_{c_2\in[0,1]}\min_{c_1\in[0,\sqrt{c_2}]}\max_{\gamma_1,\nu}  \xi_{RD,\gamma_1}(\alpha,\sigma;c_2,c_1,\nu), \label{eq:clupg5a1}
\end{equation}
where
\begin{equation}
\xi_{RD,\gamma_1}(\alpha,\sigma;c_2,c_1,\nu)=-\sqrt{c_2}+\gamma_1\sqrt{\alpha}\sqrt{1-2c_1+c_2+\sigma^2}
-\sqrt{c_2I}-\gamma_1 r-\nu c_1\sqrt{\beta}, \label{eq:clup5a2}
\end{equation}
Paralleling further the machineries of \cite{Stojnicclupint19,Stojniccluplargesc20} we consider the stationary points of the above function. Along the same lines, we first look at the derivatives with respect to $c_1$ and $c_2$ and obtain
\begin{equation}\label{eq:derc1cclupg1}
  \frac{d\xi_{RD,\gamma_1}(\alpha,\sigma;c_2,c_1,\gamma,\nu)}{d c_1}=
  -\frac{\gamma_1\sqrt{\alpha}}{\sqrt{1-2c_1+c_2+\sigma^2}}-\nu\sqrt{\beta}=0.
\end{equation}
From (\ref{eq:derc1cclupg1}) one easily finds
\begin{equation}\label{eq:derc1cclupg2}
 \nu\sqrt{\beta} = -\frac{\gamma_1\sqrt{\alpha}}{\sqrt{1-2c_1+c_2+\sigma^2}}.
\end{equation}
Computing the derivative with respect to $c_2$ we also have
\begin{equation}\label{eq:derc2cclupg1}
  \frac{d\xi_{RD,\gamma_1}(\alpha,\sigma;c_2,c_1,\nu)}{d c_2}=
  -\frac{1}{2\sqrt{c_2}}+\frac{\gamma_1\sqrt{\alpha}}{2\sqrt{1-2c_1+c_2+\sigma^2}}-\frac{\sqrt{I}}{2\sqrt{c_2}}=-\frac{1+\sqrt{I}}{2\sqrt{c_2}}-\frac{\nu\sqrt{\beta}}{2}=0,
\end{equation}
and
\begin{equation}\label{eq:derc2cclupg2}
  c_2=\left (-\frac{1+\sqrt{I}}{\nu\sqrt{\beta}}\right )^2.
\end{equation}
From (\ref{eq:derc1cclupg2}) we can then find
\begin{equation}\label{eq:derc2cclupg2a}
  c_1=\frac{1}{2}(1+c_2+\sigma^2-(\gamma_1\sqrt{\alpha}/\nu/\sqrt{\beta})^2).
\end{equation}
Now, we can switch to the derivatives with respect to $\gamma_1$ and $\nu$. We start with the derivative with respect to $\gamma_1$
\begin{equation}\label{eq:dergamma1cclupg1}
  \frac{d\xi_{RD,\gamma_1}(\alpha,\sigma;c_2,c_1,\nu)}{d\gamma_1}=\sqrt{\alpha}\sqrt{1-2c_1+c_2+\sigma^2}-r-\frac{\sqrt{c_2}}{2\sqrt{I}}\frac{dI}{d\gamma_1}.
\end{equation}
To compute the $\frac{dI}{d\gamma_1}$ we rely on (\ref{eq:clup17a}) and ultimately (\ref{eq:clup16}). In other words, we first have
\begin{equation}\label{eq:dergamma1cclupg1a}
 \frac{dI}{d\gamma_1}=\beta\left (\frac{dI_{11}(\gamma_1,\nu,c_{\ell_1})}{d\gamma_1}+\frac{dI_{12}(\gamma_1,\nu,c_{\ell_1})}{d\gamma_1}\right )
 +(1-\beta)\left (\frac{dI_{11}(\gamma_1,0,c_{\ell_1})}{d\gamma_1}+\frac{dI_{12}(\gamma_1,0,c_{\ell_1})}{d\gamma_1}\right ).
\end{equation}
Then we also have
\begin{eqnarray}\label{eq:dergamma1cclupg1b}
 \frac{dI_{11}(\gamma_1,\nu,c_{\ell_1})}{d\gamma_1} & = & \gamma_1 \erf((\nu - c_{\ell_1})/(\sqrt{2}\gamma_1))+\gamma_1 \nonumber \\
 \frac{dI_{12}(\gamma_1,\nu,c_{\ell_1})}{d\gamma_1} & = & -\gamma_1 \erf((\nu + c_{\ell_1})/(\sqrt{2}\gamma_1))+\gamma_1,
\end{eqnarray}
and analogously
\begin{eqnarray}\label{eq:dergamma1cclupg1c}
 \frac{dI_{11}(\gamma_1,0,c_{\ell_1})}{d\gamma_1} & = & \gamma_1 \erf(- c_{\ell_1}/(\sqrt{2}\gamma_1))+\gamma_1 \nonumber \\
 \frac{dI_{12}(\gamma_1,0,c_{\ell_1})}{d\gamma_1} & = & -\gamma_1 \erf( c_{\ell_1}/(\sqrt{2}\gamma_1))+\gamma_1.
\end{eqnarray}
A combination of (\ref{eq:dergamma1cclupg1})-(\ref{eq:dergamma1cclupg1c}) is then sufficient to compute the derivative with respect to $\gamma_1$. For the derivative with respect to $\nu$ we have
\begin{equation}\label{eq:dernucclupg1}
  \frac{d\xi_{RD,\gamma_1}(\alpha,\sigma;c_2,c_1,\nu)}{d\nu}=-\frac{\sqrt{c_2}}{2\sqrt{I}}\frac{dI}{d\nu}-c_1\sqrt{\beta}.
\end{equation}
To compute the $\frac{dI}{d\nu}$ we again rely on (\ref{eq:clup17a}) and ultimately (\ref{eq:clup16}). In other words, we again first have
\begin{equation}\label{eq:dernucclupg1a}
 \frac{dI}{d\nu}=\beta\left (\frac{dI_{11}(\gamma_1,\nu,c_{\ell_1})}{d\nu}+\frac{dI_{12}(\gamma_1,\nu,c_{\ell_1})}{d\nu}\right ),
\end{equation}
and
\begin{eqnarray}\label{eq:dernucclupg1b}
 \frac{dI_{11}(\gamma_1,\nu,c_{\ell_1})}{d\nu} & = & \sqrt{2/\pi}\gamma_1\exp(-(\nu - c_{\ell_1})^2/(2 \gamma_1^2)) + (\nu - c_{\ell_1})\erf((\nu - c_{\ell_1})/(\sqrt{2}\gamma_1))+(\nu-c_{\ell_1}) \nonumber \\
 \frac{dI_{12}(\gamma_1,\nu,c_{\ell_1})}{d\nu} & = & \nu+c_{\ell_1}-(\sqrt{2/\pi}\gamma_1 \exp(-(\nu + c_{\ell_1})^2/(2\gamma_1^2)) +  (\nu + c_{\ell_1}) \erf((\nu + c_{\ell_1})/(\sqrt{2}\gamma_1))).
\end{eqnarray}
The following theorem summarizes all of the above considerations.
\begin{theorem}(\bl{\textbf{CLuP -- Random dual}} -- stationary points)
 Consider $\xi_{RD,\gamma_1}(\alpha,\sigma;c_2,c_1,\nu)$ from (\ref{eq:clup5a2}). Its stationary points satisfy the following system of equations:
\begin{eqnarray}
  \frac{d\xi_{RD,\gamma_1}(\alpha,\sigma;c_2,c_1,\nu)}{d\nu} & = & -\frac{\sqrt{c_2}}{2\sqrt{I}}\beta\left (\frac{dI_{11}(\gamma_1,\nu,c_{\ell_1})}{d\nu}+\frac{dI_{12}(\gamma_1,\nu,c_{\ell_1})}{d\nu}\right )-c_1\sqrt{\beta}=0\nonumber \\
      \frac{d\xi_{RD,\gamma_1}(\alpha,\sigma;c_2,c_1,\nu)}{d\gamma_1} & = & \sqrt{\alpha}\sqrt{1-2c_1+c_2+\sigma^2}-r-\frac{\sqrt{c_2}}{2\sqrt{I}}\frac{dI}{d\gamma_1}=0 \nonumber \\
  c_2& = &\left (-\frac{1+\sqrt{I}}{\nu\sqrt{\beta}}\right )^2\nonumber \\
        c_1 & = & \frac{1}{2}(1+c_2+\sigma^2-(\gamma_1\sqrt{\alpha}/\nu/\sqrt{\beta})^2),
\end{eqnarray}
where $\frac{dI_{11}(\gamma_1,\nu,c_{\ell_1})}{d\nu}$ and $\frac{dI_{12}(\gamma_1,\nu,c_{\ell_1})}{d\nu}$ are as given in (\ref{eq:dernucclupg1b}), $I$ is given through (\ref{eq:clup16}) and (\ref{eq:clup17a}), and $\frac{dI}{d\gamma_1}$ is given through (\ref{eq:dergamma1cclupg1a})-(\ref{eq:dergamma1cclupg1c}).
\label{thm:cluprd1}
\end{theorem}
\begin{proof}
Follows through the above considerations and a long line of RDT related results presented in \cite{StojnicCSetam09,StojnicISIT2010binary,StojnicDiscPercp13,StojnicGenLasso10,StojnicGenSocp10,StojnicPrDepSocp10,StojnicRegRndDlt10,Stojnicclupint19,Stojnicclupcmpl19,Stojnicclupplt19,Stojniccluplargesc20}.
\end{proof}
\textbf{Remark} (\emph{\underline{\bl{\textbf{deterministic versus distributional priors}}}}): We should emphasize that for the concreteness and easiness of writing and following of this introductory presentation we assumed the above mentioned scenario where the elements of $\x_{sol}$ take so say concrete deterministic values (typically called the worst-case within the LASSO/SOCP context). However, the entire machinery and all the derivations presented above continue to hold with minimal changes if one assumes an arbitrary prior distribution for elements of $\x_{sol}$. If say each element of $\x_{sol}$ is a random variable $x_p$ with distribution $p(x_p)$ the key changes are: 1) in (\ref{eq:clup14}) one has only the first option for any $i$ with $\nu$ being replaced by $\nu x_p$; 2) $c_1\sqrt{\beta}$ in (\ref{eq:clup13}) is potentially scaled depending on how one wants to impose the sparsity constraint; and 3) everything is ultimately conditioned on $p(x_p)$. We will in some of our separate papers present the results that one obtains after repeating all of the above calculations with these minor changes and conditioning on $p(x_p)$ for several interesting prior distributions.

Recalling on (\ref{eq:LSclup6}) and the importance of the critical parameters $r_{sc}$, $\hat{c}_2$ , $\hat{\gamma}_1$ and $\hat{c}_{\ell_1}$ emphasized at the very beginning of this section one now can recognize the relevance of the above theorem. Namely, similarly to what was done in \cite{Stojniccluplargesc20,Stojniccluprephased20}, the above theorem provides the way to determine the values of these parameters so that the large-scale machinery defined in (\ref{eq:LSclup6}) can indeed be practically utilized. Below we will discuss the results that one can finally obtain through all of the above considerations. However, before doing that we briefly present a few results that are in a way connected to what we presented above.

%%%%%%%%%%%%%%%%%%%%%%%%%%%%%%%%%%%%%%%%%%%%%%%%%%%%%%%%%%%%%%%%%%
\subsection{Adjusted LASSO/SOCP}
\label{sec:adjlassocosp}
%%%%%%%%%%%%%%%%%%%%%%%%%%%%%%%%%%%%%%%%%%%%%%%%%%%%%%%%%%%%%%%%%%

A simple visual comparison of (\ref{eq:clup2}) and (\ref{eq:ml1socp}) is enough to conclude that there is a very strong similarity between the two optimization concepts. One would then expect that the above machinery can also be repeated starting with (\ref{eq:ml1socp}) instead of (\ref{eq:clup2}). That is of course true, with the repetition going pretty much along the same lines of what we presented above with just a tiny difference in the ending result. We skip redoing all the steps and instead summarize them through the ending results in the following LASSO/SOCP equivalent of the above Theorem \ref{thm:cluprd1}.
\begin{theorem}(\bl{\textbf{Adjusted LASSO/SOCP -- Random dual}} -- stationary points)
Similarly to (\ref{eq:clup5a2}), let
\begin{equation}
\xi_{RD,\gamma_1}^{(socp)}(\alpha,\sigma;c_2,c_1,\nu)=\gamma_1\sqrt{\alpha}\sqrt{1-2c_1+c_2+\sigma^2}
-\sqrt{c_2I}-\gamma_1 r-\nu c_1\sqrt{\beta}. \label{eq:socp5a2}
\end{equation}
Its stationary points satisfy the following system of equations:
\begin{eqnarray}
  \frac{d\xi_{RD,\gamma_1}(\alpha,\sigma;c_2,c_1,\nu)}{d\nu} & = & -\frac{\sqrt{c_2}}{2\sqrt{I}}\beta\left (\frac{dI_{11}(\gamma_1,\nu,c_{\ell_1})}{d\nu}+\frac{dI_{12}(\gamma_1,\nu,c_{\ell_1})}{d\nu}\right )-c_1\sqrt{\beta}=0\nonumber \\
      \frac{d\xi_{RD,\gamma_1}(\alpha,\sigma;c_2,c_1,\nu)}{d\gamma_1} & = & \sqrt{\alpha}\sqrt{1-2c_1+c_2+\sigma^2}-r-\frac{\sqrt{c_2}}{2\sqrt{I}}\frac{dI}{d\gamma_1}=0 \nonumber \\
  c_2& = &\left (-\frac{\sqrt{I}}{\nu\sqrt{\beta}}\right )^2\nonumber \\
        c_1 & = & \frac{1}{2}(1+c_2+\sigma^2-(\gamma_1\sqrt{\alpha}/\nu/\sqrt{\beta})^2),
\end{eqnarray}
where $\frac{dI_{11}(\gamma_1,\nu,c_{\ell_1})}{d\nu}$ and $\frac{dI_{12}(\gamma_1,\nu,c_{\ell_1})}{d\nu}$ are as given in (\ref{eq:dernucclupg1b}), $I$ is given through (\ref{eq:clup16}) and (\ref{eq:clup17a}), and $\frac{dI}{d\gamma_1}$ is given through (\ref{eq:dergamma1cclupg1a})-(\ref{eq:dergamma1cclupg1c}).
\label{thm:adjlasspsocpthm}
\end{theorem}
\begin{proof}
Follows in a way completely analogous to the way used above to derive Theorem \ref{thm:cluprd1}.
\end{proof}
A bit more experienced reader will recognize that this theorem is basically an exact analogue to the key results obtained in \cite{StojnicPrDepSocp10}. Moreover, Theorem \ref{thm:lassosocp}, which relates to what is called plain LASSO or SOCP is its a limiting worst-case variant obtained for $r=r_{socp}$ and is of course an exact analogue to some of the key results obtained in \cite{StojnicGenSocp10}. In other words, the relation between Theorems \ref{thm:lassosocp} and \ref{thm:adjlasspsocpthm} is exactly the same as the relation between some of the key results of \cite{StojnicGenSocp10} and \cite{StojnicPrDepSocp10}.

%%%%%%%%%%%%%%%%%%%%%%%%%%%%%%%%%%%%%%%%%%%%%%%%%%%%%%%%%%%%%%%%%%
\subsection{Ideal ML}
\label{sec:adjlassocosp}
%%%%%%%%%%%%%%%%%%%%%%%%%%%%%%%%%%%%%%%%%%%%%%%%%%%%%%%%%%%%%%%%%%

In this section we recall on some well known facts regarding the so-called ideal ML performance within the sparse linear regression context. Such a performance is related to an artificial algorithm that would have \emph{a priori} available knowledge where the support of $\x_{sol}$ (positions of its nonzero elements) is located. For the simplicity and without a loss of generality, we will, as earlier, assume that the first $k$ locations of $\x_{sol}$ are different from zero. We will call the part of matrix $A$ that contains the first $k$ columns $A_{(k)}$. Given all of this the original ML problem from (\ref{eq:ml1}) becomes
\begin{eqnarray}\label{eq:idealml1}
\hat{\x}^{(iml)}=\min_{\x\in\mR^k}\|\y-A_{(k)}\x\|_2.
\end{eqnarray}
The solution to this problem is trivial and is given via the so-called pseudo-inverse in the following way
\begin{eqnarray}\label{eq:idealml2}
\hat{\x}^{(iml)}=(A_{(k)}^TA_{(k)})^{-1}A_{(k)}^T\y.
\end{eqnarray}
Moreover, from (\ref{eq:linsys1}) we recall
\begin{eqnarray}\label{eq:idealml3}
\y=A_{(k)}\x_{sol}+\sigma\v.
\end{eqnarray}
Combining (\ref{eq:idealml2}) and (\ref{eq:idealml3}) we further also have
\begin{eqnarray}\label{eq:idealml4}
\hat{\x}^{(iml)}=(A_{(k)}^TA_{(k)})^{-1}A_{(k)}^T(A_{(k)}\x_{sol}+\sigma\v)=\x_{sol}+(A_{(k)}^TA_{(k)})^{-1}A_{(k)}^T\sigma\v.
\end{eqnarray}
From (\ref{eq:idealml4}) we then have for the so-called ideal ML MSE
\begin{eqnarray}\label{eq:idealml5}
\delta_{iml}=\|\hat{\x}^{(iml)}-\x_{sol}\|_2=\sigma\|(A_{(k)}^TA_{(k)})^{-1}A_{(k)}^T\v\|_2.
\end{eqnarray}
Relying on the SVD decomposition of $A_{(k)}=U \Sigma V^T$ ($U^TU$ and $V^TV$ are the identity matrices and $\Sigma$ is the diagonal matrix with its diagonal elements being the so-called singular values of matrix $A_{(k)}$) we can transform (\ref{eq:idealml5}) into
\begin{equation}\label{eq:idealml6}
\delta_{iml}=\sigma\|(A_{(k)}^TA_{(k)})^{-1}A_{(k)}^T\v\|_2=\sigma\|(V\Sigma^2V^T)^{-1}(U \Sigma V^T)^T\v\|_2
=\sigma\|V\Sigma^{-2}V^TV \Sigma U^T\v\|_2=\sigma\|V\Sigma^{-1}U^T\v\|_2. \\
\end{equation}
Given the rotational invariance of $\v$ and independence between $\v$ and $A_{(k)}$ we also have
\begin{equation}\label{eq:idealml7}
\delta_{iml}=\sigma\|V\Sigma^{-1}U\v\|_2=\sigma\sum_{i=1}^k\frac{\v_i^2}{\Sigma_i^2}.
\end{equation}
Using the law of large numbers, independence between $\v$ and $\Sigma$, and the Marchenko-Pastur distribution of the eigenvalues of Wigner matrices we obtain
\begin{equation}\label{eq:idealml8}
\lim_{n\rightarrow\infty}\delta_{iml}= \sigma\sqrt{\frac{\beta}{2\pi\alpha}\int_{\lambda_-}^{\lambda_+}\frac{\sqrt{(\lambda_+-x)(x-\lambda_-)}}{\lambda_{sp} x^2}dx},
\end{equation}
where
\begin{equation}\label{eq:idealml9}
\lambda_{sp}=\frac{\beta}{\alpha}, \qquad \lambda_+=(1+\sqrt{\lambda_{sp}})^2, \qquad \lambda_-=(1-\sqrt{\lambda_{sp}})^2.
\end{equation}
We will use these results below as the absolute benchmark for the ML sparse regression.

%%%%%%%%%%%%%%%%%%%%%%%%%%%%%%%%%%%%%%%%%%%%%%%%%%%%%%%%%%%%%%%%%%
\subsection{Numerical results}
\label{sec:nummlspreg}
%%%%%%%%%%%%%%%%%%%%%%%%%%%%%%%%%%%%%%%%%%%%%%%%%%%%%%%%%%%%%%%%%%

In this section we will present the numerical results that can be obtained through the mechanisms discussed above. They will relate to both, the theoretical predictions and the simulations. We start below with some of the theoretical predictions.

%%%%%%%%%%%%%%%%%%%%%%%%%%%%%%%%%%%%%%%%%%%%%%%%%%%%%%%%%%%%%%%%%%
\subsubsection{Theoretical predictions}
\label{sec:numtheory}
%%%%%%%%%%%%%%%%%%%%%%%%%%%%%%%%%%%%%%%%%%%%%%%%%%%%%%%%%%%%%%%%%%

As mentioned earlier, the results of Theorem \ref{thm:cluprd1} can be used to characterize the performance of the introduced CLuP mechanism. One though has to be careful how to interpret these results. Before discussing all these intricacies we should emphasize that there are many parameters that can be used to characterize the CLuP's performance. For the concreteness and clarity, we here focus on the MSE (we will discuss quite a few other ones on various occasions in separate papers and often even within a different context where they may be of more interest).

Now, the most natural way to use Theorem \ref{thm:cluprd1} would be roughly the following. One can choose a pair $(r_{sc},c_{\ell_1})$ and for such a pair find the corresponding quantities that Theorem \ref{thm:cluprd1} provides. Among them would be $c_2$ and $c_1$. Using such $c_2$ and $c_1$ one can then compute $\delta$. Redoing this procedure until one finds the pair $(r_{sc},c_{\ell_1})$ that gives the minimal $\delta$ would lead to the \textbf{very ultimate} CLuP theoretically predictable performance. In mathematical terminology, one has
\begin{equation}\label{eq:numtheory1}
\lim_{n\rightarrow\infty}\delta_{clup}^{(vult)}=\lim_{n\rightarrow\infty}\min_{r_{sc},c_{\ell_1}}\|\x_{sol}-\x\|_2=\min_{r_{sc},c_{\ell_1}}\sqrt{1-2c_1+c_2},
\end{equation}
with $c_2$ and $c_1$ being functions of $(r_{sc},c_{\ell_1})$ as stated in Theorem \ref{thm:cluprd1}. While this type of choice for $(r_{sc},c_{\ell_1})$ goes for the minimal MSE, one has to be careful as to how easy/difficult would be for the CLuP to actually achieve the corresponding stationary point where such minimal MSE happens. A bit more careful method would then be to not necessarily look for the minimal MSE but rather for the minimal MSE where getting to the corresponding stationary point is not jeopardized by being stuck in other stationary points. We refer to such performance as the \textbf{ultimate} CLuP performance.

For the concreteness and to ensure the easiness of following, we will choose the same scenario for which we will compare various theoretical predictions and simulated results. We will assume moderately under-sampled regime with $\alpha=0.5$. We will also assume $\beta=0.1625$. It is not that difficult to double check that this corresponds to the scenario from Theorem \ref{thm:lassosocp} where $\alpha_w=0.45$. Moreover, a little bit of familiarity with the LASSO and SOCP theoretical predictions from \cite{StojnicGenLasso10,StojnicGenSocp10} says that this choice further corresponds to the scenario that is fairly close to the so-called phase-transition regime where the plain LASSO or SOCP have the worst-case residual MSE
\begin{equation}\label{eq:numtheory2}
\lim_{n\rightarrow\infty}\delta_{lasso}=\lim_{n\rightarrow\infty}\delta_{socp}=\sigma\sqrt{\frac{\alpha_w}{\alpha-\alpha_w}}=3\sigma.
\end{equation}
Also, for this very same scenario we have from (\ref{eq:idealml8}) and (\ref{eq:idealml9})
\begin{equation}\label{eq:numtheory3}
\lim_{n\rightarrow\infty}\delta_{iml}= \sigma\sqrt{\frac{\beta}{2\pi\alpha}\int_{\lambda_-}^{\lambda_+}\frac{\sqrt{(\lambda_+-x)(x-\lambda_-)}}{\lambda_{sp} x^2}dx}=0.6939\sigma.
\end{equation}

For the parameters that we set above, we in Figure \ref{fig:mmselassoclupverybest} show the results that can be obtained based on Theorems \ref{thm:lassosocp} and \ref{thm:cluprd1} and (\ref{eq:idealml8}) and (\ref{eq:idealml9}). As mentioned above, our focus is on the MSE (mean square error) $\delta$ as a performance measure (MSE is of course not only our choice here, but typically the most often used quantity to characterize/measure the quality of the underlying regression mechanism). As can be seen from the figure, the very ultimate CLuP performance substantially improves over the LASSO from (\ref{eq:ml1lasso}) and the SOCP from (\ref{eq:ml1socp}). Moreover, it also comes fairly close to the so-called ideal ML performance in a solid range of SNRs. The ideal ML performance is typically thought of as a very hard to break barrier as it relies on perfect knowledge of the location of nonzero components of $\x$ (when it comes to the ML concept it is actually unbreakable). As mentioned earlier, the results of Theorem \ref{thm:lassosocp} are essentially the limiting worst-case version of the results one would get through Theorem \ref{thm:adjlasspsocpthm}. We should also add that in the range of $\sigma$ shown in Figure \ref{fig:mmselassoclupverybest} the adjustment to the results of Theorem \ref{thm:lassosocp} that Theorem \ref{thm:adjlasspsocpthm} provides are very marginal and basically invisible (they start showing effects in the lower SNR regimes for $1/\sigma$ below $6$). In Table \ref{tab:tabclup1}, we show in parallel the numerical values that correspond to the data shown in Figure \ref{fig:mmselassoclupverybest}.
\begin{figure}[htb]
%\begin{minipage}[b]{.5\linewidth}
\centering
\centerline{\epsfig{figure=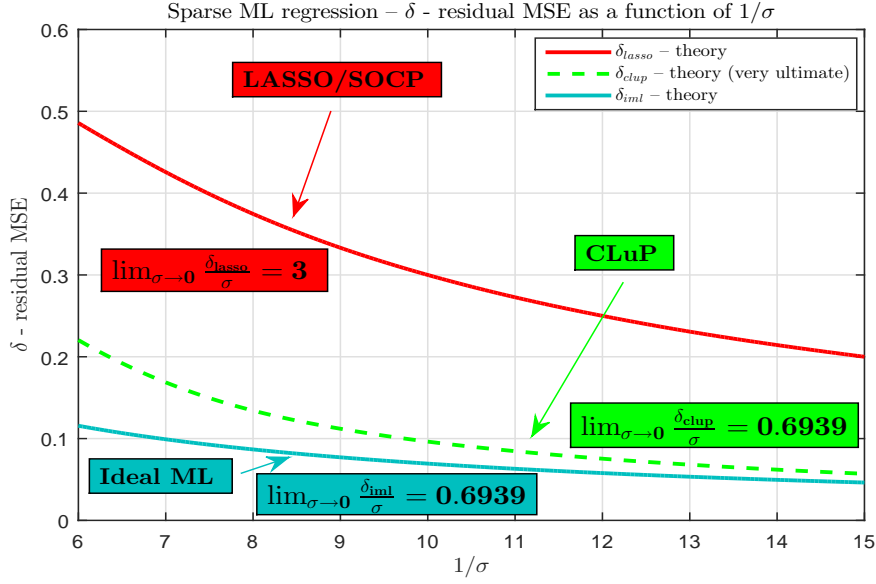,width=13.5cm,height=8cm}}
%\end{minipage}
%\begin{minipage}[b]{.5\linewidth}
%\centering
%\centerline{\epsfig{figure=finprerral08.eps,width=9cm,height=6.5cm}}
%\end{minipage}
\caption{Comparison of $\delta$ as a function of $1/\sigma$; $\alpha=0.5$; $\beta=0.1625$; very ultimate level ($n\rightarrow \infty$)}
\label{fig:mmselassoclupverybest}
\end{figure}

\begin{table}[h]
\caption{CLuP -- \textbf{theoretical} values for  $c_{\ell_1}  $, $r_{sc}  $, $c_2$,  $c_1$, $\xi_{RD}$, and $\delta$ (very ultimate)} \vspace{.1in}
\hspace{-0in}\centering
\footnotesize{
\begin{tabular}{||c||c|c||c|c||c|c||c||}\hline\hline
$ 1/\sigma $ & $c_{\ell_1}$ & $r_{sc}  $  & $c_2$ &  $c_1$ & $\xi_{RD}$ & $\delta$ & $\delta/\sigma$\\ \hline\hline
$\mathbf{6 }$ & $\mathbf{3.0661 }$ & $\mathbf{1.5568 }$ & $\mathbf{0.9533 }$ & $\mathbf{0.9524 }$ & $\mathbf{0.2846 }$ & $\mathbf{0.2204 }$ & $\mathbf{1.3224 }$ \\ \hline
$\mathbf{7 }$ & $\mathbf{2.8162 }$ & $\mathbf{2.0436 }$ & $\mathbf{0.9715 }$ & $\mathbf{0.9715 }$ & $\mathbf{0.1496 }$ & $\mathbf{0.1687 }$ & $\mathbf{1.1809 }$ \\ \hline
$\mathbf{8 }$ & $\mathbf{2.7160 }$ & $\mathbf{2.1558 }$ & $\mathbf{0.9820 }$ & $\mathbf{0.9820 }$ & $\mathbf{0.1039 }$ & $\mathbf{0.1342 }$ & $\mathbf{1.0736 }$ \\ \hline
$\mathbf{9 }$ & $\mathbf{2.6615 }$ & $\mathbf{2.2272 }$ & $\mathbf{0.9875 }$ & $\mathbf{0.9875 }$ & $\mathbf{0.0792 }$ & $\mathbf{0.1120 }$ & $\mathbf{1.0080 }$ \\ \hline
$\mathbf{10 }$ & $\mathbf{2.6269 }$ & $\mathbf{2.2777 }$ & $\mathbf{0.9907 }$ & $\mathbf{0.9907 }$ & $\mathbf{0.0636 }$ & $\mathbf{0.0963 }$ & $\mathbf{0.9630 }$ \\ \hline
$\mathbf{11 }$ & $\mathbf{2.6033 }$ & $\mathbf{2.3155 }$ & $\mathbf{0.9928 }$ & $\mathbf{0.9928 }$ & $\mathbf{0.0530 }$ & $\mathbf{0.0846 }$ & $\mathbf{0.9306 }$ \\ \hline
$\mathbf{12 }$ & $\mathbf{2.5857 }$ & $\mathbf{2.3448 }$ & $\mathbf{0.9943 }$ & $\mathbf{0.9943 }$ & $\mathbf{0.0452 }$ & $\mathbf{0.0754 }$ & $\mathbf{0.9048 }$ \\ \hline
$\mathbf{13 }$ & $\mathbf{2.5728 }$ & $\mathbf{2.3684 }$ & $\mathbf{0.9954 }$ & $\mathbf{0.9954 }$ & $\mathbf{0.0394 }$ & $\mathbf{0.0681 }$ & $\mathbf{0.8853 }$ \\ \hline
$\mathbf{14 }$ & $\mathbf{2.5624 }$ & $\mathbf{2.3877 }$ & $\mathbf{0.9961 }$ & $\mathbf{0.9961 }$ & $\mathbf{0.0348 }$ & $\mathbf{0.0620 }$ & $\mathbf{0.8680 }$ \\ \hline
$\mathbf{15 }$ & $\mathbf{2.5541 }$ & $\mathbf{2.4039 }$ & $\mathbf{0.9967 }$ & $\mathbf{0.9967 }$ & $\mathbf{0.0312 }$ & $\mathbf{0.0570 }$ & $\mathbf{0.8550 }$ \\ \hline\hline
$\mathbf{100 }$ & $\mathbf{2.4871 }$ & $\mathbf{2.5699 }$ & $\mathbf{0.9999 }$ & $\mathbf{0.9999 }$ & $\mathbf{0.0026 }$ & $\mathbf{0.0072 }$ & $\mathbf{0.7171 }$ \\ \hline\hline
$\mathbf{\infty }$ (limit)& $\mathbf{2.4807 }$ & $\mathbf{2.5981 }$ & $\mathbf{1}$ & $\mathbf{1}$ & $\mathbf{0 }$ & $\mathbf{0 }$ & $\mathbf{0.6939 }$ \\ \hline
\hline
\end{tabular}}
\label{tab:tabclup1}
\end{table}

%%%%%%%%%%%%%%%%%%%%%%%%%%%%%%%%%%%%%%%%%%%%%%%%%%%%%%%%%%%%%%%%%%
\subsubsection{Achieving the ideal ML}
\label{sec:achievingidealml}
%%%%%%%%%%%%%%%%%%%%%%%%%%%%%%%%%%%%%%%%%%%%%%%%%%%%%%%%%%%%%%%%%%

Looking a bit more carefully at the results presented in Table \ref{tab:tabclup1} one observes that in addition to the values that correspond to the data shown in Figure \ref{fig:mmselassoclupverybest}, it also contains (in the last two rows) limiting values as $\sigma$ approaches zero. From the very last element in the last row one also observes that it matches the value given in (\ref{eq:numtheory3}). This actually means that the CLuP is ultimately capable of achieving the ideal ML performance. Below we sketch the arguments that confirm that the limiting MSE of the above discussed CLuP indeed matches the corresponding one of the ideal ML given in (\ref{eq:numtheory3}).

There are many ways how this can be achieved. We choose a particular way that is not necessarily the fastest but highlights the limiting behavior of the entire machinery and its all relevant performance characterizing parameters discussed above. We start with several simple observations. Let
\begin{eqnarray}
  c_{1s} &\triangleq & \frac{1-c_{1}}{\sigma^2} \nonumber \\
  c_{2s} &\triangleq & \frac{1-c_{2}}{\sigma^2} \nonumber \\
  \nu_{s} &\triangleq & \nu\sigma \nonumber \\
  c_{\ell_1s} &\triangleq & \frac{c_{\ell_1}\sqrt{\beta}-1}{\sigma} .\label{eq:achidml1}
\end{eqnarray}
From (\ref{eq:dernucclupg1a}) and (\ref{eq:dernucclupg1b}) we have as $\sigma\rightarrow 0$
\begin{equation}\label{eq:achidml2}
  \frac{dI}{d\nu}\longrightarrow 2\beta(\nu+c_{\ell_1}).
\end{equation}
From (\ref{eq:dernucclupg1}) we then also have
\begin{eqnarray}\label{eq:achidml3}
 & &  \frac{d\xi_{RD,\gamma_1}(\alpha,\sigma;c_2,c_1,\nu)}{d\nu}  =  0\nonumber \\
\Longleftrightarrow & &  -\frac{\sqrt{c_2}}{2\sqrt{I}}\frac{dI}{d\nu}-c_1\sqrt{\beta} =0 \nonumber \\
\Longleftrightarrow & &  \sqrt{I}\longrightarrow -\frac{\sqrt{c_2}(\nu+c_{\ell_1})\sqrt{\beta}}{c_1}.
\end{eqnarray}
Recalling on (\ref{eq:derc2cclupg1}) we can write
\begin{equation}\label{eq:achidml4}
-\frac{1+\sqrt{I}}{2\sqrt{c_2}}-\frac{\nu\sqrt{\beta}}{2}=0 \quad \Longleftrightarrow \quad
\sqrt{c_2}=-\frac{1+\sqrt{I}}{\nu\sqrt{\beta}}.
\end{equation}
Plugging the value for $\sqrt{I}$ from (\ref{eq:achidml3}) into (\ref{eq:achidml4}) gives
\begin{equation}\label{eq:achidml5}
\sqrt{c_2}=-\frac{1+\sqrt{I}}{\nu\sqrt{\beta}}=-\frac{1}{\nu\sqrt{\beta}}-\frac{\sqrt{I}}{\nu\sqrt{\beta}}
 \quad \Longleftrightarrow \quad
 \sqrt{c_2}\longrightarrow-\frac{1}{\nu\sqrt{\beta}}-\frac{-\frac{\sqrt{c_2}(\nu+c_{\ell_1})\sqrt{\beta}}{c_1}}{\nu\sqrt{\beta}},
\end{equation}
and after a few additional transformations
\begin{equation}\label{eq:achidml6}
\nu\sqrt{\beta}\longrightarrow\frac{c_{\ell_1}\sqrt{\beta}\sqrt{c_2}-c_1}{\sqrt{c_2}(c_1-1)}
\longrightarrow\frac{c_{\ell_1}\sqrt{\beta}-1}{c_1-1}.
\end{equation}
Combining (\ref{eq:achidml1}) and (\ref{eq:achidml6}) we arrive at the following
\begin{equation}\label{eq:achidml7}
|\nu_s|\sqrt{\beta}\longrightarrow\frac{c_{\ell_1s}}{c_{1s}}.
\end{equation}
Recalling on (\ref{eq:derc1cclupg2}) we have
\begin{equation}\label{eq:achidml8}
 |\nu|\sqrt{\beta} = \frac{\gamma_1\sqrt{\alpha}}{\sqrt{1-2c_1+c_2+\sigma^2}},
\end{equation}
and combining with (\ref{eq:achidml1}) we also have
\begin{equation}\label{eq:achidml9}
 |\nu_s|\sqrt{\beta} \longrightarrow \frac{\gamma_1\sqrt{\alpha}}{\sqrt{2c_{1s}-c_{2s}+1}}.
\end{equation}
Now, set
\begin{equation}\label{eq:achidml10}
F=(1-\beta)((\gamma_1^2+1/\beta)\erfc(1/\gamma_1/\sqrt{2\beta})-2\gamma_1/\sqrt{\beta}/\sqrt{2\pi}\exp(-(1/2/\beta/\gamma_1^2))),
\end{equation}
and
\begin{equation}\label{eq:achidml11}
D=2\beta\nu c_{\ell_1}+\beta\gamma_1^2+\beta c_{\ell_1}^2+F.
\end{equation}
From (\ref{eq:clup16}) and (\ref{eq:clup17a}) we find
\begin{equation}\label{eq:achidml12}
I\longrightarrow\beta\nu^2+D.
\end{equation}
Moreover, based on (\ref{eq:achidml12}), we can further write for $\sqrt{I}$
\begin{equation}\label{eq:achidml13}
\sqrt{I}\longrightarrow\sqrt{\beta\nu^2}\sqrt{1+\frac{D}{\beta\nu^2}}.
\end{equation}
Taylor expansion gives
\begin{equation}\label{eq:achidml14}
\sqrt{I}\longrightarrow\sqrt{\beta\nu^2}\sqrt{1+\frac{D}{\beta\nu^2}}\longrightarrow\sqrt{\beta\nu^2}\left (1+\frac{D}{2\beta\nu^2}
-\frac{D^2}{8(\beta\nu^2)^2}\right )=\sqrt{\beta\nu^2}+\frac{D}{2\sqrt{\beta\nu^2}}
-\frac{D^2}{8\sqrt{\beta\nu^2}^3}.
\end{equation}
Combining (\ref{eq:achidml11}) and (\ref{eq:achidml14}) we have
\begin{eqnarray}\label{eq:achidml15}
\sqrt{I}& \longrightarrow &\sqrt{\beta\nu^2}+\frac{2\beta\nu c_{\ell_1}+\beta\gamma_1^2+\beta c_{\ell_1}^2+F}{2\sqrt{\beta\nu^2}}
-\frac{(2\beta\nu c_{\ell_1})^2}{8\sqrt{\beta\nu^2}^3}\nonumber \\
& \longrightarrow &\sqrt{\beta\nu^2}-\sqrt{\beta} c_{\ell_1}+\frac{\beta\gamma_1^2+F}{2\sqrt{\beta\nu^2}}
+\frac{\beta c_{\ell_1}^2}{2\sqrt{\beta\nu^2}}
-\frac{(2\beta\nu c_{\ell_1})^2}{8\sqrt{\beta\nu^2}^3}\nonumber \\
& \longrightarrow &\sqrt{\beta\nu^2}-\sqrt{\beta} c_{\ell_1}+\frac{\beta\gamma_1^2+F}{2\sqrt{\beta\nu^2}}.
\end{eqnarray}
Moreover, combining  (\ref{eq:achidml15}) with (\ref{eq:achidml1}) gives
\begin{eqnarray}\label{eq:achidml16}
\sqrt{I}& \longrightarrow & \frac{\sqrt{\beta}|\nu_s|}{\sigma}-1- c_{\ell_1s}\sigma+\frac{\beta\gamma_1^2+F}{2\sqrt{\beta}|\nu_s|}\sigma.
\end{eqnarray}
From (\ref{eq:achidml3}) we also have
\begin{eqnarray}\label{eq:achidml17}
  \sqrt{I}\longrightarrow -\frac{\sqrt{c_2}(\nu+c_{\ell_1})\sqrt{\beta}}{c_1}=-\frac{\sqrt{1-(1-c_2)}(\nu+c_{\ell_1})\sqrt{\beta}}{1-(1-c_1)},
\end{eqnarray}
and
\begin{eqnarray}\label{eq:achidml18}
  \sqrt{I} & \longrightarrow & -\frac{\sqrt{1-(1-c_2)}(\nu+c_{\ell_1})\sqrt{\beta}}{1-(1-c_1)} \nonumber \\
  & = & \frac{\sqrt{1-(1-c_2)}|\nu|\sqrt{\beta}}{1-(1-c_1)}
  -\frac{\sqrt{1-(1-c_2)}c_{\ell_1}\sqrt{\beta}}{1-(1-c_1)}\nonumber \\
  & = & \frac{\sqrt{1-c_{2s}\sigma^2}|\nu|\sqrt{\beta}}{1-c_{1s}\sigma^2}
  -\frac{\sqrt{1-c_{2s}\sigma^2}c_{\ell_1}\sqrt{\beta}}{1-c_{1s}\sigma^2}\nonumber \\
  & \longrightarrow & \left (1-\frac{c_{2s}}{2}\sigma^2\right )(1+c_{1s}\sigma^2)|\nu|\sqrt{\beta}
  -\left (1-\frac{c_{2s}}{2}\sigma^2\right )(1+c_{1s}\sigma^2)c_{\ell_1}\sqrt{\beta}.
  \end{eqnarray}
After a few additional transformations and neglecting further terms of order $\sigma^2$ or smaller one finally has
\begin{eqnarray}\label{eq:achidml19}
  \sqrt{I}   & \longrightarrow & \left (1-\frac{c_{2s}}{2}\sigma^2\right )(1+c_{1s}\sigma^2)|\nu|\sqrt{\beta}
  -\left (1-\frac{c_{2s}}{2}\sigma^2\right )(1+c_{1s}\sigma^2)c_{\ell_1}\sqrt{\beta} \nonumber \\
     & \longrightarrow & |\nu|\sqrt{\beta}+\left (c_{1s}-\frac{c_{2s}}{2}\right )|\nu|\sqrt{\beta}\sigma^2
 -c_{\ell_1}\sqrt{\beta} -\left (c_{1s}-\frac{c_{2s}}{2}\right )\sigma^2c_{\ell_1}\sqrt{\beta} \nonumber \\
      & \longrightarrow & \frac{|\nu_s|\sqrt{\beta}}{\sigma}+\left (c_{1s}-\frac{c_{2s}}{2}\right )|\nu_s|\sqrt{\beta}\sigma
 -1-c_{\ell_1s}\sigma.
  \end{eqnarray}
Comparing (\ref{eq:achidml19}) to (\ref{eq:achidml16}) we obtain
\begin{eqnarray}\label{eq:achidml20}
\frac{\beta\gamma_1^2+F}{2\sqrt{\beta}|\nu_s|}= \left (c_{1s}-\frac{c_{2s}}{2}\right )|\nu_s|\sqrt{\beta}.
  \end{eqnarray}
Now, we set
\begin{eqnarray}\label{eq:achidml21}
A\triangleq (\beta\gamma_1^2+F)= (2c_{1s}-c_{2s})|\nu_s|^2\beta,
  \end{eqnarray}
and through a utilization of (\ref{eq:achidml8}) continue to obtain
\begin{eqnarray}\label{eq:achidml22}
A= (2c_{1s}-c_{2s}+1)|\nu_s|^2\beta-|\nu_s|^2\beta=\gamma_1^2\alpha-|\nu_s|^2\beta.
  \end{eqnarray}
From (\ref{eq:achidml22}) one then trivially has
\begin{eqnarray}\label{eq:achidml23}
|\nu_s|^2\beta=\gamma_1^2\alpha-A.
  \end{eqnarray}
Recalling on (\ref{eq:dergamma1cclupg1})
\begin{eqnarray}\label{eq:achidml24}
  & & \frac{d\xi_{RD,\gamma_1}(\alpha,\sigma;c_2,c_1,\nu)}{d\gamma_1}=0\nonumber \\
 \Longleftrightarrow & &  \sqrt{\alpha}\sqrt{1-2c_1+c_2+\sigma^2}-r-\frac{\sqrt{c_2}}{2\sqrt{I}}\frac{dI}{d\gamma_1}=0.
\end{eqnarray}
From (\ref{eq:dergamma1cclupg1a})--(\ref{eq:dergamma1cclupg1c}) one finds
\begin{eqnarray}\label{eq:achidml25}
\frac{dI}{d\gamma_1} \longrightarrow 2\gamma_1\beta+2(1-\beta)\gamma_1\erfc(c_{\ell_1}/(\sqrt{2}\gamma_1)).
\end{eqnarray}
A combination of (\ref{eq:achidml24}) and (\ref{eq:achidml25}) gives
\begin{eqnarray}\label{eq:achidml26}
& &  \sqrt{\alpha}\sqrt{1-2c_1+c_2+\sigma^2}-r-\frac{\sqrt{c_2}}{2\sqrt{I}}\frac{dI}{d\gamma_1} \longrightarrow 0 \nonumber \\
\Longleftrightarrow & & \sqrt{\alpha}\sqrt{2c_{1s}-c_{2s}+1}\sigma-r_{sc}r_{socp}-\frac{\sigma}{|\nu_s|\sqrt{\beta}}
(\gamma_1\beta+(1-\beta)\gamma_1\erfc(c_{\ell_1}/(\sqrt{2}\gamma_1))) \longrightarrow 0.
\end{eqnarray}
Utilizing (\ref{eq:achidml9}) and (\ref{eq:achidml23}) we find
\begin{eqnarray}\label{eq:achidml27}
& &  \sqrt{\alpha}\sqrt{1-2c_1+c_2+\sigma^2}-r-\frac{\sqrt{c_2}}{2\sqrt{I}}\frac{dI}{d\gamma_1} \longrightarrow 0 \nonumber \\
\Longleftrightarrow & & \frac{\gamma_1\alpha}{|\nu_s|\sqrt{\beta}}\sigma-r_{sc}\frac{r_{socp}}{\sqrt{n}}-\frac{\sigma}{|\nu_s|\sqrt{\beta}}
(\gamma_1\beta+(1-\beta)\gamma_1\erfc(c_{\ell_1}/(\sqrt{2}\gamma_1))) \longrightarrow 0 \nonumber \\
\Longleftrightarrow & & \gamma_1\alpha\sigma-r_{sc}\frac{r_{socp}}{\sqrt{n}}|\nu_s|\sqrt{\beta}-\sigma
(\gamma_1\beta+(1-\beta)\gamma_1\erfc(c_{\ell_1}/(\sqrt{2}\gamma_1))) \longrightarrow 0\nonumber \\
\Longleftrightarrow & & \gamma_1\alpha\sigma-r_{sc}\frac{r_{socp}}{\sqrt{n}}\sqrt{\gamma_1^2\alpha-A}-
\sigma(\gamma_1\beta+(1-\beta)\gamma_1\erfc(c_{\ell_1}/(\sqrt{2}\gamma_1))) \longrightarrow 0.
\end{eqnarray}
Recalling on the choice $r_{socp}$ from Theorem \ref{thm:lassosocp} one has
\begin{eqnarray}\label{eq:achidml28}
& &  \sqrt{\alpha}\sqrt{1-2c_1+c_2+\sigma^2}-r-\frac{\sqrt{c_2}}{2\sqrt{I}}\frac{dI}{d\gamma_1} \longrightarrow 0 \nonumber \\
\Longleftrightarrow & & \gamma_1\alpha\sigma-r_{sc}\sigma\sqrt{\alpha-\alpha_w}\sqrt{\gamma_1^2\alpha-A}-\sigma
(\gamma_1\beta+(1-\beta)\gamma_1\erfc(c_{\ell_1}/(\sqrt{2}\gamma_1))) \longrightarrow 0 \nonumber \\
\Longleftrightarrow & & \gamma_1\alpha-r_{sc}\sqrt{\alpha-\alpha_w}\sqrt{\gamma_1^2\alpha-A}-
(\gamma_1\beta+(1-\beta)\gamma_1\erfc(c_{\ell_1}/(\sqrt{2}\gamma_1))) =0.
\end{eqnarray}
Choosing $\gamma_1\rightarrow 0$ from (\ref{eq:achidml10}) and (\ref{eq:achidml21}) we have
\begin{eqnarray}\label{eq:achidml29}
F\rightarrow 0 \quad \mbox{and} \quad A\rightarrow \gamma_1^2\beta.
\end{eqnarray}
Moreover, we then find from (\ref{eq:achidml28})
\begin{eqnarray}\label{eq:achidml30}
& &  \sqrt{\alpha}\sqrt{1-2c_1+c_2+\sigma^2}-r-\frac{\sqrt{c_2}}{2\sqrt{I}}\frac{dI}{d\gamma_1} \longrightarrow 0 \nonumber \\
\Longleftrightarrow & & \gamma_1\alpha-r_{sc}\sqrt{\alpha-\alpha_w}\sqrt{\gamma_1^2\alpha-\gamma_1^2\beta}-
(\gamma_1\beta+(1-\beta)\gamma_1\erfc(c_{\ell_1}/(\sqrt{2}\gamma_1))) =0 \nonumber \\
\Longleftrightarrow & & \gamma_1\alpha-r_{sc}\sqrt{\alpha-\alpha_w}\sqrt{\gamma_1^2\alpha-\gamma_1^2\beta}-
\gamma_1\beta =0 \nonumber \\
\Longleftrightarrow & & \alpha-r_{sc}\sqrt{\alpha-\alpha_w}\sqrt{\alpha-\beta}-\beta =0.
\end{eqnarray}
Finally from (\ref{eq:achidml30}) one has for the optimal choice of $r_{sc}$
\begin{eqnarray}\label{eq:achidml31}
r_{sc}^{(opt)}=\sqrt{\frac{\alpha-\beta}{\alpha-\alpha_w}}.
\end{eqnarray}
Also, a combination of (\ref{eq:achidml9}) and (\ref{eq:achidml23}) gives
\begin{eqnarray}\label{eq:achidml32}
\sqrt{2c_{1s}-c_{2s}+1}=\frac{\gamma_1\sqrt{\alpha}}{\gamma_1^2\alpha-A},
\end{eqnarray}
and
\begin{eqnarray}\label{eq:achidml33}
\sqrt{2c_{1s}-c_{2s}}=\sqrt{\frac{A}{\gamma_1^2\alpha-A}}.
\end{eqnarray}
Recalling on the definition of the MSE $\delta$ we further have
\begin{eqnarray}\label{eq:achidml34}
\lim_{n\rightarrow\infty}\delta=\lim_{n\rightarrow\infty}\sqrt{1-2c_1+c2}\longrightarrow\sigma\sqrt{2c_{1s}-c_{2s}}.
\end{eqnarray}
Finally combining (\ref{eq:achidml33}) and (\ref{eq:achidml34}) we find the very ultimate CLuP's residual MSE as $\sigma\rightarrow 0$
\begin{eqnarray}\label{eq:achidml35}
\lim_{n\rightarrow\infty}\delta_{clup}\longrightarrow\sigma\sqrt{2c_{1s}-c_{2s}}=\sigma\sqrt{\frac{A}{\gamma_1^2\alpha-A}}=\sigma\sqrt{\frac{\beta}{\alpha-\beta}},
\end{eqnarray}
or in a more convenient form
\begin{eqnarray}\label{eq:achidml36}
\lim_{\sigma\rightarrow 0}\lim_{n\rightarrow\infty}\frac{\delta_{clup}}{\sigma}=\sqrt{\frac{\beta}{\alpha-\beta}}.
\end{eqnarray}
For particular values $\alpha=0.5$ and $\beta=0.1625$ that we considered earlier, we also have
\begin{eqnarray}\label{eq:achidml37}
\lim_{\sigma\rightarrow 0}\lim_{n\rightarrow\infty}\frac{\delta_{clup}}{\sigma}=\sqrt{\frac{\beta}{\alpha-\beta}}=0.6939.
\end{eqnarray}
It is of course not that hard to see that this is exactly the same as what one gets for the ideal ML after solving the integral in (\ref{eq:numtheory3}). For the completeness we also mention that from the above derivation and (\ref{eq:achidml31}) and (\ref{eq:achidml1}) we have
\begin{eqnarray}\label{eq:achidml38}
r_{sc}^{(opt)}=\sqrt{\frac{\alpha-\beta}{\alpha-\alpha_w}}=2.5981,
\end{eqnarray}
and
\begin{eqnarray}\label{eq:achidml39}
c_{\ell_1}^{(opt)}=\frac{1}{\sqrt{\beta}}=2.4807.
\end{eqnarray}
These values are exactly matching the corresponding values in the last row of Table \ref{tab:tabclup1}. Moreover, as one can observe from the second to last row of the table, for $\sigma=0.01$ ($1/\sigma=100$) one has that the values for $c_{\ell_1}$, $r_{sc}$, and $\delta/\sigma$ are very close to the ultimate limiting ones obtained assuming $\sigma\rightarrow 0$.

%%%%%%%%%%%%%%%%%%%%%%%%%%%%%%%%%%%%%%%%%%%%%%%%%%%%%%%%%%%%%%%%%%
\subsubsection{Practial CLuP}
\label{sec:practicalclup}
%%%%%%%%%%%%%%%%%%%%%%%%%%%%%%%%%%%%%%%%%%%%%%%%%%%%%%%%%%%%%%%%%%

As achieving the above mentioned very ultimate CLuP performance might require a bit more advanced approach we in Figure \ref{fig:mmselassoclupbest} show the results that can be obtained based on Theorem \ref{thm:cluprd1} but at the same time can also be approached through practical realizations (the above mentioned ultimate level of performance). They are obtained for particular choices for $(r_{sc},c_{\ell_1})$ and in Table \ref{tab:tabclup2} we show in parallel the explicit values for the $(r_{sc},c_{\ell_1})$ that correspond to the data in Figure \ref{fig:mmselassoclupbest}. As discussed earlier, we refer to the resulting CLuP performance as the ultimate CLuP performance. Since the optimal $r_{sc}$ is fairly close to $2$ across all considered $\sigma$ values we opted for a bit of sub-optimality and fixed $r_{sc}=2$. This is convenient from two points of view: 1) while it is sub-obtimal it is still very close to the optimal choice and the resulting MSE is larger by an almost negligible value than the one that can be obtained if one insists on the $r_{sc}$'s optimality; and 2) having $r_{sc}$ being fixed across a range of $\sigma$ values is particularly important when it comes to practical running the underlying CLuP as it doesn't require separate tuning due to potential SNR changes.

\begin{figure}[htb]
%\begin{minipage}[b]{.5\linewidth}
\centering
\centerline{\epsfig{figure=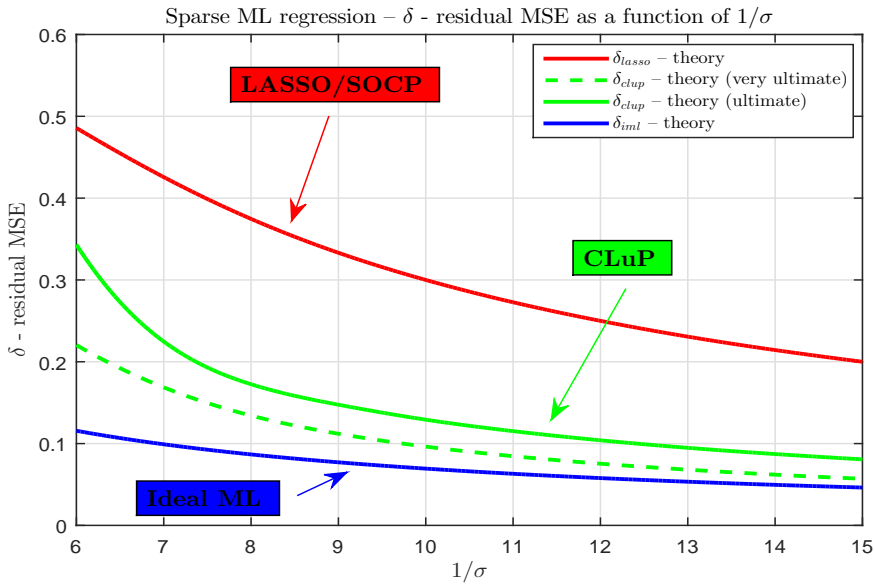,width=13.5cm,height=8cm}}
%\end{minipage}
%\begin{minipage}[b]{.5\linewidth}
%\centering
%\centerline{\epsfig{figure=finprerral08.eps,width=9cm,height=6.5cm}}
%\end{minipage}
\caption{Comparison of $\delta$ as a function of $1/\sigma$; $\alpha=0.5$; $\beta=0.1625$; practical favorable level ($n\rightarrow \infty$)}
\label{fig:mmselassoclupbest}
\end{figure}

\begin{table}[h]
\caption{CLuP -- \textbf{theoretical} values for $c_{\ell_1}  $, $r_{sc}  $,  $c_2$,  $c_1$, $\xi_{RD}$, and $\delta$ (ultimate)} \vspace{.1in}
\hspace{-0in}\centering
\footnotesize{
\begin{tabular}{||c||c|c||c|c||c|c||}\hline\hline
$ 1/\sigma $ & $c_{\ell_1}  $ & $r_{sc}  $ &  $c_2$ &  $c_1$ & $\xi_{RD}$ & $\delta$\\ \hline\hline
$\mathbf{6 }$ & $\mathbf{5.05 }$ & $\mathbf{2}$ & $\mathbf{0.7957 }$ & $\mathbf{0.8511 }$ & $\mathbf{0.9392 }$ & $\mathbf{0.3059 }$ \\ \hline
$\mathbf{7 }$ & $\mathbf{4.54 }$ & $\mathbf{2}$ & $\mathbf{0.8464 }$ & $\mathbf{0.9009 }$ & $\mathbf{0.8016 }$ & $\mathbf{0.2114 }$ \\ \hline
$\mathbf{8 }$ & $\mathbf{4.37 }$ & $\mathbf{2}$ & $\mathbf{0.8723 }$ & $\mathbf{0.9218 }$ & $\mathbf{0.7500 }$ & $\mathbf{0.1693 }$ \\ \hline
$\mathbf{9 }$ & $\mathbf{4.27 }$ & $\mathbf{2}$ & $\mathbf{0.8901 }$ & $\mathbf{0.9349 }$ & $\mathbf{0.7171 }$ & $\mathbf{0.1426 }$ \\ \hline
$\mathbf{10 }$ & $\mathbf{4.22 }$ & $\mathbf{2}$ & $\mathbf{0.9026 }$ & $\mathbf{0.9436 }$ & $\mathbf{0.7007 }$ & $\mathbf{0.1239 }$ \\ \hline
$\mathbf{11 }$ & $\mathbf{4.17 }$ & $\mathbf{2}$ & $\mathbf{0.9130 }$ & $\mathbf{0.9505 }$ & $\mathbf{0.6828 }$ & $\mathbf{0.1095 }$ \\ \hline
$\mathbf{12 }$ & $\mathbf{4.14 }$ & $\mathbf{2}$ & $\mathbf{0.9210 }$ & $\mathbf{0.9557 }$ & $\mathbf{0.6719 }$ & $\mathbf{0.0984 }$ \\ \hline
$\mathbf{13 }$ & $\mathbf{4.12 }$ & $\mathbf{2}$ & $\mathbf{0.9275 }$ & $\mathbf{0.9598 }$ & $\mathbf{0.6646 }$ & $\mathbf{0.0894 }$ \\ \hline
$\mathbf{14 }$ & $\mathbf{4.10 }$ & $\mathbf{2}$ & $\mathbf{0.9332 }$ & $\mathbf{0.9632 }$ & $\mathbf{0.6570 }$ & $\mathbf{0.0819 }$ \\ \hline
$\mathbf{15 }$ & $\mathbf{4.09 }$ & $\mathbf{2}$ & $\mathbf{0.9378 }$ & $\mathbf{0.9660 }$ & $\mathbf{0.6533 }$ & $\mathbf{0.0757 }$ \\ \hline
\hline
\end{tabular}}
\label{tab:tabclup2}
\end{table}

Looking at the results presented in Figure \ref{fig:mmselassoclupbest} and Table \ref{tab:tabclup2} one can see that they are trailing by a tiny margin the corresponding ones from Figure \ref{fig:mmselassoclupverybest} and Table \ref{tab:tabclup1}. Moreover, as we will see bellow the values shown in Figure \ref{fig:mmselassoclupbest} and Table \ref{tab:tabclup2} can indeed be approached through the practical running. On the other hand, a similar type of conclusion can not necessarily be made for the very ultimate ones from Figure \ref{fig:mmselassoclupverybest} and Table \ref{tab:tabclup1}. To fully understand the source of the difference between these two sets of results one would need to have an excellent level of understanding of all the intricacies discussed in \cite{Stojnicclupint19,Stojnicclupcmpl19,Stojnicclupplt19,Stojniccluprephased20,Stojniccluplargesc20}. A thorough discussion in this direction goes well beyond the scope of the present paper. However, in a separate paper we will discuss this gap in greater details and provide avenues that eventually lead to bridging the gap.

%%%%%%%%%%%%%%%%%%%%%%%%%%%%%%%%%%%%%%%%%%%%%%%%%%%%%%%%%%%%%%%%%%
\subsubsection{Simulations}
\label{sec:numsim}
%%%%%%%%%%%%%%%%%%%%%%%%%%%%%%%%%%%%%%%%%%%%%%%%%%%%%%%%%%%%%%%%%%

We have also conducted quite a few numerical experiments. In Figure \ref{fig:mmselassoclupsim} and Table \ref{tab:tabclup3} we show the results that we obtained through them. All the key parameters are as before, i.e. $\alpha=0.5$ and $\beta=0.1625$, and $c_2$ and $\gamma_1$ are determined through the above machinery. There are just few tiny differences though. To ensure better concentrations and consequently an overall smooth large-scale CLuP running we backed off a little bit from the above ultimate choice for $c_{\ell_1}$. We selected $c_{\ell_1}=4.5$ over pretty much the entire SNR range. The sole exception was $1/\sigma=7$ where we selected $c_{\ell_1}=5$. Also, we selected $r_{sc}=2$ for any choice of $\sigma$. It turns out that these almost universal choices are working rather well. One achieves the MSE performance that is fairly close to the above discussed ultimate CLuP and at the same time doesn't need to tune these parameters as functions of $\sigma$. We chose a moderately large $n=2000$ and $c_{q,2}=7\sqrt{n}$ as a solid starting point with an option for occasional periodic increase by a couple of percent after say $50$ iterations. We typically fixed the maximum number of iterations to be $3000$ (much less though often sufficed) and ran the simplest version of the above mentioned large-scale CLuP without any of the sophisticated rerunning options from \cite{Stojniccluprephased20,Stojniccluplargesc20}.
\begin{figure}[htb]
%\begin{minipage}[b]{.5\linewidth}
\centering
\centerline{\epsfig{figure=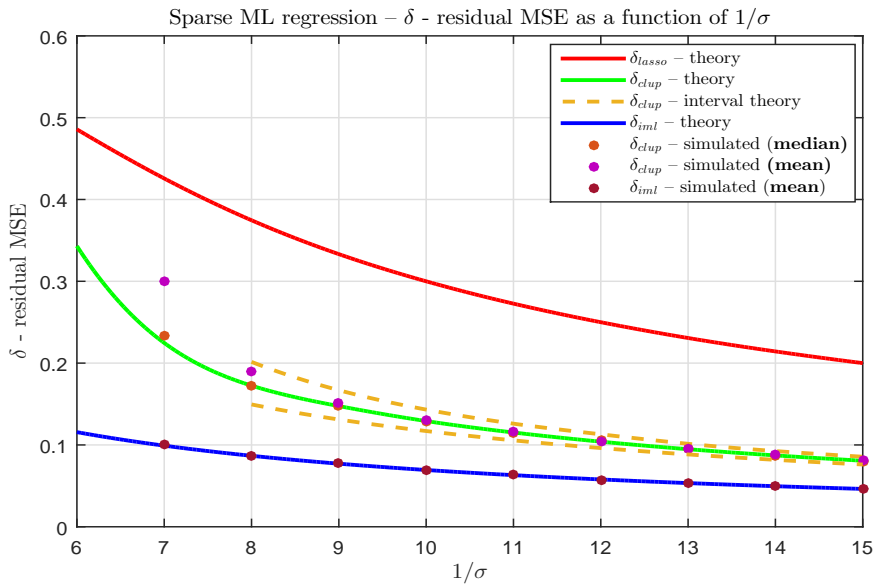,width=13.5cm,height=8cm}}
%\end{minipage}
%\begin{minipage}[b]{.5\linewidth}
%\centering
%\centerline{\epsfig{figure=finprerral08.eps,width=9cm,height=6.5cm}}
%\end{minipage}
\caption{Comparison of $\delta$ as a function of $1/\sigma$; $\alpha=0.5$; $\beta=0.1625$; simulations and interval predictions}
\label{fig:mmselassoclupsim}
\end{figure}
\begin{table}[h]
\caption{CLuP -- \textbf{theoretical/\bl{simulated}} values for $c_2$,  $c_1$, $\xi_{RD}$, and $\delta$ ($r_{sc}=2$; $c_{\ell_1}=4.5$ (for $ 1/\sigma=7$, $c_{\ell_1}=5$ )); $n=2000$} \vspace{.1in}
\hspace{-0in}\centering
\footnotesize{
\begin{tabular}{||c||c||c|c||c|c||c|c||c|c|c||}\hline\hline
$ 1/\sigma $ & $\gamma_1$ &  $c_2$ & $c_2$ & $c_1$ & $c_1$ &  $\xi_{RD}$ & $\xi_{RD}$ & $\delta$ & $\delta$ (mean) & $\delta$ (median)\\ \hline\hline
$\mathbf{7 }$ & $\mathbf{3.5060 }$ & $\mathbf{0.8273 }$ & $\bl{\mathbf{0.8025 }}$ & $\mathbf{0.8884 }$ & $\bl{\mathbf{0.8473 }}$ & $\mathbf{0.9720 }$ & $\bl{\mathbf{1.0353 }}$ & $\mathbf{0.2247 }$ & $\bl{\mathbf{0.2997 }}$ & $\bl{\mathbf{0.2341 }}$ \\ \hline
$\mathbf{8 }$ & $\mathbf{3.0630 }$ & $\mathbf{0.8670 }$ & $\bl{\mathbf{0.8579 }}$ & $\mathbf{0.9186 }$ & $\bl{\mathbf{0.9081 }}$ & $\mathbf{0.7998 }$ & $\bl{\mathbf{0.8360 }}$ & $\mathbf{0.1725 }$ & $\bl{\mathbf{0.1907 }}$ & $\bl{\mathbf{0.1722 }}$ \\ \hline
$\mathbf{9 }$ & $\mathbf{3.0664 }$ & $\mathbf{0.8815 }$ & $\bl{\mathbf{0.8788 }}$ & $\mathbf{0.9299 }$ & $\bl{\mathbf{0.9271 }}$ & $\mathbf{0.8064 }$ & $\bl{\mathbf{0.8257 }}$ & $\mathbf{0.1475 }$ & $\bl{\mathbf{0.1517 }}$ & $\bl{\mathbf{0.1474 }}$ \\ \hline
$\mathbf{10 }$ & $\mathbf{3.0703 }$ & $\mathbf{0.8931 }$ & $\bl{\mathbf{0.8899 }}$ & $\mathbf{0.9382 }$ & $\bl{\mathbf{0.9364 }}$ & $\mathbf{0.8102 }$ & $\bl{\mathbf{0.8344 }}$ & $\mathbf{0.1292 }$ & $\bl{\mathbf{0.1302 }}$ & $\bl{\mathbf{0.1289 }}$ \\ \hline
$\mathbf{11 }$ & $\mathbf{3.0741 }$ & $\mathbf{0.9026 }$ & $\bl{\mathbf{0.8993 }}$ & $\mathbf{0.9447 }$ & $\bl{\mathbf{0.9428 }}$ & $\mathbf{0.8125 }$ & $\bl{\mathbf{0.8331 }}$ & $\mathbf{0.1152 }$ & $\bl{\mathbf{0.1163 }}$ & $\bl{\mathbf{0.1152 }}$ \\ \hline
$\mathbf{12 }$ & $\mathbf{3.0775 }$ & $\mathbf{0.9106 }$ & $\bl{\mathbf{0.9065 }}$ & $\mathbf{0.9499 }$ & $\bl{\mathbf{0.9478 }}$ & $\mathbf{0.8141 }$ & $\bl{\mathbf{0.8370 }}$ & $\mathbf{0.1040 }$ & $\bl{\mathbf{0.1046 }}$ & $\bl{\mathbf{0.1053 }}$ \\ \hline
$\mathbf{13 }$ & $\mathbf{3.0805 }$ & $\mathbf{0.9173 }$ & $\bl{\mathbf{0.9147 }}$ & $\mathbf{0.9542 }$ & $\bl{\mathbf{0.9528 }}$ & $\mathbf{0.8151 }$ & $\bl{\mathbf{0.8310 }}$ & $\mathbf{0.0948 }$ & $\bl{\mathbf{0.0955 }}$ & $\bl{\mathbf{0.0946 }}$ \\ \hline
$\mathbf{14 }$ & $\mathbf{3.0832 }$ & $\mathbf{0.9231 }$ & $\bl{\mathbf{0.9201 }}$ & $\mathbf{0.9578 }$ & $\bl{\mathbf{0.9562 }}$ & $\mathbf{0.8158 }$ & $\bl{\mathbf{0.8319 }}$ & $\mathbf{0.0872 }$ & $\bl{\mathbf{0.0879 }}$ & $\bl{\mathbf{0.0874 }}$ \\ \hline
$\mathbf{15 }$ & $\mathbf{3.0856 }$ & $\mathbf{0.9282 }$ & $\bl{\mathbf{0.9269 }}$ & $\mathbf{0.9608 }$ & $\bl{\mathbf{0.9602 }}$ & $\mathbf{0.8162 }$ & $\bl{\mathbf{0.8228 }}$ & $\mathbf{0.0807 }$ & $\bl{\mathbf{0.0806 }}$ & $\bl{\mathbf{0.0804 }}$ \\ \hline
\hline
\end{tabular}}
\label{tab:tabclup3}
\end{table}
As can be seen from both,  Figure \ref{fig:mmselassoclupsim} and Table \ref{tab:tabclup3}, the agreement between the theoretical predictions and the values obtained through the simulations for all four key critical system parameters $c_2$,  $c_1$, $\xi_{RD}$, and $\delta$  is rather solid. There are very tiny differences in the low SNR regime $1/\sigma=7$ or $1/\sigma=8$. This is due to occasional appearance of bad instances. We also provided the median values for the MSE. As can be seen from Table \ref{tab:tabclup3} even in these regimes the median values are almost exactly matching the theoretical predictions. For the completeness we also simulated the ideal ML performance. As can be seen from Figure \ref{fig:mmselassoclupsim}, the results we obtained for the ideal ML are almost indistinguishable from the corresponding ones obtained through the theoretical predictions, i.e. through (\ref{eq:numtheory3}). We should also add that one can define $r_{sc},c_{\ell_1}$ dependent limiting MSE as
\begin{equation}\label{eq:simeq1}
  \lim_{\sigma\rightarrow 0}\frac{\delta(r_{sc},c_{\ell_1})}{\sigma}.
\end{equation}
Repeating earlier calculations one can find that for the above choice $r_{sc}=2$ and $c_{\ell_1}=4.5$ the limiting MSE is
\begin{equation}\label{eq:simeq1}
  \lim_{\sigma\rightarrow 0}\lim_{n\rightarrow\infty}\frac{\delta(2,4.5)}{\sigma}=1.1178.
\end{equation}
This is slightly above the very ultimate value $0.6939$ discussed earlier but still fairly close to it and almost three times smaller than the corresponding value $3$ that one obtains from Theorem \ref{thm:lassosocp} for LASSO/SOCP.

Carefully looking at what we presented in Figure \ref{fig:mmselassoclupsim}, one can observe that in addition to the already discussed results related to the LASSO/SOCP, the CLuP, and the ideal ML we also presented the theoretical predictions for the given scenarios based on the so-called \textbf{\emph{interval approach}} (light brown dashed curve). This approach is relatively basic but as can be seen from Figure \ref{fig:mmselassoclupsim} gives pretty good bounding interval for the resulting MSE. Below, we will briefly sketch how one can obtain these interval predictions.

%%%%%%%%%%%%%%%%%%%%%%%%%%%%%%%%%%%%%%%%%%%%%%%%%%%%%%%%%%%%%%%%%%
\subsubsection{Interval approach predictions}
\label{sec:numinterval}
%%%%%%%%%%%%%%%%%%%%%%%%%%%%%%%%%%%%%%%%%%%%%%%%%%%%%%%%%%%%%%%%%%

All the main ingredients needed for this approach are already present in our earlier discussions in the previous sections as well as in a long line of work \cite{Stojnicclupint19,Stojnicclupcmpl19,Stojnicclupplt19,Stojniccluplargesc20,Stojniccluprephased20,StojnicCSetam09,StojnicCSetamBlock09,StojnicISIT2010binary,StojnicDiscPercp13,StojnicUpper10,StojnicGenLasso10,StojnicGenSocp10,StojnicPrDepSocp10,StojnicRegRndDlt10,Stojnicbinary16fin,Stojnicbinary16asym}. Here we only formalize a procedure that can be utilized to obtain the results shown in Figure \ref{fig:mmselassoclupsim}.

We start by viewing $c_1$ as a running parameter and continue by considering the following optimization
\begin{equation}
\min_{c_2}\max_{\gamma_1,\nu}(\xi_{RD,\gamma_1}(\alpha,\sigma;c_2,c_1,\nu)+\sqrt{c_2}). \label{eq:clupinterval1}
\end{equation}
Let the functions ${c}_{2f}(c_1)$, $\gamma_{1f}(c_1)$, and $\nu_f(c_1)$ be the solution triplet of this optimization. Utilizing the machineries of \cite{StojnicGenLasso10,StojnicGenSocp10,StojnicPrDepSocp10} and the strong random duality one then has
\begin{equation}
\xi_{ub}\triangleq\min_{c_1}(\xi_{RD,\gamma_{1f}(c_1)}(\alpha,\sigma;c_{2f}(c_1),c_1,\nu_f(c_1)))\geq \xi_p. \label{eq:clupinterval2}
\end{equation}
It then trivially follows that $\xi_{ub}\geq \max_{\gamma_1,\nu}\xi_{RD,\gamma_1}(\alpha,\sigma;c_2,c_1,\nu)$ and one can then establish the intervals where $\delta_{clup}$ resides with overwhelming probability as
\begin{equation}\label{eq:clupinterval3}
  \delta_{clup}\in[\delta_{lb},\delta_{ub}],
\end{equation}
with
\begin{eqnarray}
\delta_{lb} & = & \min_{c_1,c_2} \sqrt{1-2c_1+c_2}\nonumber \\
\mbox{subject to} & & |\xi_{ub}-\max_{\gamma_1,\nu}\xi_{RD,\gamma_1}(\alpha,\sigma;c_2,c_1,\nu)|=0, \label{eq:clupinterval4}
\end{eqnarray}
and
\begin{eqnarray}
\delta_{ub} & = & \max_{c_1,c_2} \sqrt{1-2c_1+c_2}\nonumber \\
\mbox{subject to} & & |\xi_{ub}-\max_{\gamma_1,\nu}\xi_{RD,\gamma_1}(\alpha,\sigma;c_2,c_1,\nu)|=0. \label{eq:clupinterval5}
\end{eqnarray}
We summarize this mechanism in the following theorem.
\begin{theorem}(\bl{\textbf{CLuP}} -- simple interval approach)
  Let $\xi_{RD,\gamma_1}(\alpha,\sigma;c_2,c_1,\nu)$ be as in (\ref{eq:clup5a2}). Also, let ${c}_{2f}(c_1)$, $\gamma_{1f}(c_1)$, and $\nu_f(c_1)$ be the triplet of functions that are the solution to the optimization problem in (\ref{eq:clupinterval1}). Set
  \begin{equation}
\xi_{ub}\triangleq\min_{c_1}(\xi_{RD,\gamma_{1f}(c_1)}(\alpha,\sigma;c_{2f}(c_1),c_1,\nu_f(c_1))). \label{eq:thmclupinterval2}
\end{equation}
One then has
\begin{equation}\label{eq:thmclupinterval3}
 \lim_{n\rightarrow\infty}P(\delta_{clup}\triangleq\|\x_{sol}-\x^{(clup)}\|_2\in[\delta_{lb},\delta_{ub}])\rightarrow 1,
\end{equation}
where $\delta_{lb}$ and $\delta_{ub}$ are as given in (\ref{eq:clupinterval4}) and (\ref{eq:clupinterval5}), respectively.
\label{thm:thmclupinterval}
\end{theorem}
\begin{proof}
  Follows through the above considerations and ultimately a collection of our earlier results from \cite{Stojnicclupint19,Stojnicclupcmpl19,Stojnicclupplt19,Stojniccluplargesc20,Stojniccluprephased20,StojnicCSetam09,StojnicCSetamBlock09,StojnicISIT2010binary,StojnicDiscPercp13,StojnicUpper10,StojnicGenLasso10,StojnicGenSocp10,StojnicPrDepSocp10,StojnicRegRndDlt10,Stojnicbinary16fin,Stojnicbinary16asym}.
\end{proof}
In Table \ref{tab:tabclup4} we show the numerical results that one can obtain based on Theorem \ref{thm:thmclupinterval}. These results exactly correspond to the data shown as light brown dashed curve in Figure \ref{fig:mmselassoclupsim}. As can be seen from both, Figure \ref{fig:mmselassoclupsim} and Table \ref{tab:tabclup4}, the intervals where $\delta_{clup}$ is located are fairly narrow and quite close to the ultimate predictions obtained based on Theorem \ref{thm:cluprd1}. Moreover, as $\sigma$ goes down the widths of the interval shrink.
\begin{table}[h]
\caption{CLuP -- \textbf{theoretical/\bl{simulated}} values for $\delta$ together with the intervals $[\delta_{lb},\delta_{ub}]$ ($r_{sc}=2$; $c_{\ell_1}=4.5$; $n=2000$} \vspace{.1in}
\hspace{-0in}\centering
\footnotesize{
\begin{tabular}{||c||c||c|c|c||c||}\hline\hline
$ 1/\sigma $ & $\delta_{lb}$ & $\delta$ & $\delta$ (mean) & $\delta$ (median) & $\delta_{ub}$\\ \hline\hline
$\mathbf{8 }$ & $\ultclupcola{\mathbf{0.1495 }}$ & $\mathbf{0.1725 }$ & $\bl{\mathbf{0.1907 }}$ & $\bl{\mathbf{0.1722 }}$ & $\ultclupcola{\mathbf{0.2015 }}$ \\ \hline
$\mathbf{9 }$ & $\ultclupcola{\mathbf{0.1311 }}$ & $\mathbf{0.1475 }$ & $\bl{\mathbf{0.1517 }}$ & $\bl{\mathbf{0.1474 }}$ & $\ultclupcola{\mathbf{0.1669 }}$ \\ \hline
$\mathbf{10 }$ & $\ultclupcola{\mathbf{0.1170 }}$ & $\mathbf{0.1292 }$ & $\bl{\mathbf{0.1302 }}$ & $\bl{\mathbf{0.1289 }}$ & $\ultclupcola{\mathbf{0.1432 }}$ \\ \hline
$\mathbf{11 }$ & $\ultclupcola{\mathbf{0.1056 }}$ & $\mathbf{0.1152 }$ & $\bl{\mathbf{0.1163 }}$ & $\bl{\mathbf{0.1152 }}$ & $\ultclupcola{\mathbf{0.1260 }}$ \\ \hline
$\mathbf{12 }$ & $\ultclupcola{\mathbf{0.0962 }}$ & $\mathbf{0.1040 }$ & $\bl{\mathbf{0.1046 }}$ & $\bl{\mathbf{0.1053 }}$ & $\ultclupcola{\mathbf{0.1130 }}$ \\ \hline
$\mathbf{13 }$ & $\ultclupcola{\mathbf{0.0884 }}$ & $\mathbf{0.0948 }$ & $\bl{\mathbf{0.0955 }}$ & $\bl{\mathbf{0.0946 }}$ & $\ultclupcola{\mathbf{0.1014 }}$ \\ \hline
$\mathbf{14 }$ & $\ultclupcola{\mathbf{0.0818 }}$ & $\mathbf{0.0872 }$ & $\bl{\mathbf{0.0879 }}$ & $\bl{\mathbf{0.0874 }}$ & $\ultclupcola{\mathbf{0.0925 }}$ \\ \hline
$\mathbf{15 }$ & $\ultclupcola{\mathbf{0.0761 }}$ & $\mathbf{0.0807 }$ & $\bl{\mathbf{0.0806 }}$ & $\bl{\mathbf{0.0804 }}$ & $\ultclupcola{\mathbf{0.0856 }}$ \\ \hline
\hline
\end{tabular}}
\label{tab:tabclup4}
\end{table}

%%%%%%%%%%%%%%%%%%%%%%%%%%%%%%%%%%%%%%%%%%%%%%%%%%%%%%%%%%%%%%%%%%%%%%%%%%%%%%%%%%
\section{Conclusion}
\label{sec:conc}
%%%%%%%%%%%%%%%%%%%%%%%%%%%%%%%%%%%%%%%%%%%%%%%%%%%%%%%%%%%%%%%%%%%%%%%%%%%%%%%%%%

In this paper we consider the so-called sparse ML regression problem. It is among the most fundamental estimation problems and as such appears in many scientific fields, with statistics, machine learning, information theory, linear estimation, and signal processing probably being mathematically the most prominent ones. We introduced a novel algorithmic mechanism (that we refer to as CLuP) for handling these kinds of ML problems. The algorithm is based on some of our recent works \cite{Stojnicclupint19,Stojnicclupcmpl19,Stojnicclupplt19,Stojniccluplargesc20,Stojniccluprephased20} related to the MIMO ML detection and a long line of our work related to the \bl{\textbf{Random duality theory (RDT)}} that we systematically developed in \cite{StojnicCSetam09,StojnicCSetamBlock09,StojnicISIT2010binary,StojnicDiscPercp13,StojnicUpper10,StojnicGenLasso10,StojnicGenSocp10,StojnicPrDepSocp10,StojnicRegRndDlt10,Stojnicbinary16fin,Stojnicbinary16asym}.

As an RDT mechanism, the CLuP has some of the typical RDT excellent features. One can precisely characterize its performance and at the same time utilize the underlying analysis to create its a convenient practical implementation. Moreover, we presented a particular implementation that is well tailored for the so-called large-scale scenarios. Such scenarios are expected to dominate the big data era and are of special interest in the context of statistical regression, estimation, data analysis, and prediction. Although very simple, the implementation that we introduced can handle with an ease problems with several thousands of unknowns. In fact, its computational complexity per iteration is \textbf{\emph{quadratic}} and as such theoretically minimal possible and includes literally only a single matrix-vector multiplication. Moreover, when it comes to potential implementations on problems with numbers of variables reaching several hundreds of thousands or millions one can not hope to have a more efficient per iteration algorithm.

In addition to providing precise performance characterizations and ensuring excellent large scale capabilities, we also observe that the characterized performances themselves are also very favorable. We discussed several types of performances, from the very ultimate one to the practically more realizable ones. The introduced CLuP mechanism on its a very ultimate level can even achieve the exact ideal ML (the ideal ML is an unbreakable barrier within the ML context as it relies on the exact knowledge (which of course typically is not available in regression or any other ML considerations) of the unknown vector of parameters' support). Its very ultimate residual ML is also significantly smaller (often more than $4$ times) than the corresponding one of the state of the art LASSO or SOCP alternatives. In addition to the very ultimate CLuP's implementation we also looked at more practical realizations and demonstrated that they come fairly close to the very ultimate level while maintaining excellent computational complexity properties and being able to easily handle problems with several thousands of unknowns. We accompanied our theoretical predictions with a solid set of results obtained through numerical simulations. As is typically the case for any of our random duality algorithmic considerations, the observed agreement between the theoretical predictions and the values obtained through the numerical experiments is very strong.

There are many different avenueas for building further. We mentioned quite a few of them throughout the presentation itself. In particular, given that this is the introductory paper on CLuP's handling of sparse regression, we opted for a bit of simplicity and showcased a CLuP's performance that is a bit below the very ultimate one. In some of our separate papers we will discuss in great details more advanced strategies that one can utilize to actually approach the very ultimate performance. Also, the presented technical machinery can easily handle other types of regressions. Moreover, it can easily be adapted to handle similar types of regression that appear in different scientific fields as well. Sometimes, all these applications require a bit of technical modifications but are basically a routine procedure once the core mechanisms presented here and in \cite{Stojnicclupint19,Stojnicclupcmpl19,Stojnicclupplt19,Stojniccluplargesc20,Stojniccluprephased20} (and earlier in \cite{StojnicCSetam09,StojnicCSetamBlock09,StojnicISIT2010binary,StojnicDiscPercp13,StojnicUpper10,StojnicGenLasso10,StojnicGenSocp10,StojnicPrDepSocp10,StojnicRegRndDlt10,Stojnicbinary16fin,Stojnicbinary16asym}) are available. As mentioned earlier, for some of the most interesting related problems we will in separate papers show how these modifications can be done and what kind of results one eventually can get through them.

\begin{singlespace}
\bibliographystyle{plain}
\bibliography{clupmlspregRefs}
\end{singlespace}

\end{document}